\documentclass[a4paper, 12pt]{article}

	\usepackage[francais,english]{babel} 				
	\usepackage[utf8]{inputenc} 				
	\usepackage{amsmath} 						
	\usepackage{amsthm} 						
	\usepackage{amssymb}
	\usepackage{makeidx} 						
	\usepackage{bbm}

	\usepackage{color}

	\usepackage{url}							
	\usepackage{subfigure}
	\usepackage{graphicx}
	\usepackage{multicol}
	\usepackage{pifont}

			\usepackage{algorithm}
			\usepackage{algorithmic}
			
			\usepackage{fullpage}
			\usepackage{authblk}

\begin{document}		

	\definecolor{dkgreen}{rgb}{0,0.6,0}
	\definecolor{gray}{rgb}{0.5,0.5,0.5}
	\definecolor{mauve}{rgb}{0.58,0,0.82}



\theoremstyle{plain}			
	\newtheorem{lem}		{Lemme}
	\newtheorem{prop}		{Proposition}
	\newtheorem{theo}		{Theorem}	
	\newtheorem{coro}		{Corollaire}
\theoremstyle{definition}
	\newtheorem{defn}		{Definition}
	\newtheorem{conj}		{Conjecture}
	\newtheorem{exmp}		{Exemple}
\theoremstyle{remark}			
	\newtheorem*{case}		{Cas}	
	\newtheorem*{mynote}		{Note}	
	\newtheorem{remark}		{Remark}	







 \title{Incompressible limit of a continuum model of tissue growth with segregation for two cell populations}

 \author[1]{Alina Chertock}
 \author[2]{Pierre Degond}
  \author[2,3]{Sophie Hecht}
   \author[3]{Jean-Paul Vincent}
   \affil[1]{Department of Mathematics, North Carolina State University, Raleigh, NC 27695, USA }
   \affil[2]{Department of Mathematics, Imperial College London, London SW7~2AZ, UK }
   \affil[3]{Francis Crick Institute, 1 Midland Rd, London NW1 1AT, UK }
   \affil[ ]{\textit {chertock@math.ncsu.edu, p.degond@imperial.ac.uk, sh5015@ic.ac.uk, jp.vincent@crick.ac.uk}}



\maketitle



\begin{abstract}
This paper proposes a model for the growth two interacting populations of cells that do not mix. The dynamics is driven by pressure and cohesion forces on the one hand and proliferation on the other hand. Following earlier works on the single population case, we show that the model approximates a free boundary Hele Shaw type model that we characterise using both analytical and numerical arguments. \end{abstract}

\providecommand{\Acknowledgements}[1]{\textbf{{Acknowledgements. }} #1}
\Acknowledgements{The work of AC was supported in part by NSF Grant DMS-521051 and RNMS Grant DMS-1107444 (KI-Net). SH and JPV acknowledge support from the Francis Crick Institute which receives its core funding from Cancer Research UK (FC001204), the UK Medical Research Council (FC001204), and the Wellcome Trust (FC001204). PD acknowledges support by the Engineering and Physical Sciences 
Research Council (EPSRC) under grants no. EP/M006883/1, EP/N014529/1 and 
EP/P013651/1, by the Royal 
Society and the Wolfson Foundation through a Royal Society Wolfson 
Research Merit Award no. WM130048 and by the National Science 
Foundation (NSF) under grant no. RNMS11-07444 (KI-Net). PD is on leave 
from CNRS, Institut de Math\'ematiques de Toulouse, France. PD and SH would like to thank Andrea Bertozzi for stimulating discussions. }

\providecommand{\DataStatement}[1]{\textbf{{Data Statement. }} #1}
\DataStatement{No new data were collected in the course of this research.}

\providecommand{\keywords}[1]{\textbf{{Keywords.}} #1}

\providecommand{\AMS}[1]{\textbf{{AMS subject classifications.}} #1}
\AMS{35K55; 35R35; 65M08; 92C15; 92C10}

\section{Introduction}

During development, tissues and organs grow while generating diverse cell types. Thus, different cell populations co-exist during growth. For example, in a developing limb, prospective muscles, bone and epidermis become distinct during development and as a result grow at different rates. As they grow, these cells types contribute mass gain for the whole structure. However, since the different cell populations grow at different rates, stress arises and must be relieved to ensure that they contribute harmoniously to the final structure. How differential growth rates within a structure are accommodated is therefore an important question in developmental biology. To approach this question from a theoretical point of view, we consider a model whereby an idealised tissue is composed of one contiguous cell population located within a wider area occupied by another cell population.

Our approach relies on a new continuum model for two populations of cells which includes the following biological features. First, we impose a constraint that cells do not overlap which gives rise to maximal packing tissue density. This is ensured in the model by an appropriate pressure-density relationship which becomes singular at the packing density. Second, we model cell-cell contact inhibition by implementing cell motion in the direction opposite to the cell density gradient. These features  are believed to play important roles in the development of  mono-layered epithelial tissues such as pseudo-stratified epithelia. Third, the model incorporates features that are specific to two non mixing cell populations. The model favors segregation by penalizing the mixing of cells of different populations. Such segregation is observed in various tissues, such as developing tissues when there are two populations of cells which are genetically distinct, or cancerous tissues composed of proliferative and healthy cells. Our model is set in an optimal transport framework. We describe the cell populations by means of continuum densities and assume that they satisfy coupled gradient flow equations which tend to decrease an appropriate free energy. The free energy decay encompasses important properties of the biological system and the use of a gradient flow ensures that this decay is built into the model. The choice of the metric used to compute the free energy gradient is also a critical modelling choice. In this paper, we choose the Wasserstein metric as it is the natural choice to measure the distances between densities and many different types of partial differential equation (PDE) have been derived as gradient flows in this metric (examples include the porous medium equation \cite{Otto}, parabolic of convection-diffusion equations \cite{KW}, the Fokker Planck equation \cite{JKO}, etc). 

After describing the model, we investigate its incompressible limit. In this limit, the the possible values of the cell densities are reduced to either zero or their respective maximal value corresponding to the packing density. We show that this limit model is a free boundary Hele-Shaw model (HSM), which allows one to focus on the geometric evolution of the boundaries of the domain occupied by the two species.

Mathematical models have been widely used to study tissue development or tumour growth. Among these, we distinguish two ways of representing the cells. On the one hand, discrete models consider each cell as an individual entity whose position and other attributes evolve in time \cite{DH,GLD}. This provides a high level of accuracy but also results in large computational costs. On the other hand, continuum models consider local averages such as the cell number density, as function of space and time, which evolve according to a suitable PDE \cite{AM04,BCGRS}. This description is appropriate when the number of cells is large as it dramatically reduces the computational cost. However it only gives access to the large scale features of the system as the small scale features are averaged out. As the goal of this paper is to study the evolution of the whole tissue, we have chosen to use a continuum model.  Continuum models roughly fall into two categories. The first category comprises models which describe the dynamics of the cell density through convection and diffusion \cite{Byrne,BenAmar,RBEJPJ}. In the second category, we find models which describe the motion of the geometric boundary between the tissue or the tumour and its environment through geometric evolution equations \cite{Cui,FriedmanHu,Greenspan,Lowengrub}. The latter share similarities with Stefan's free boundary problem in solidification \cite{HMS}. These two types of models are related to one another through asymptotic limits. In particular, some tumour growth models of the first type have been related to Hele-Shaw free boundary models in \cite{PQV,PQTV,PV}.

In this paper, we propose a two cell population model as an extension of earlier models for single cell populations introduced in \cite{BD,PQV} in the context of tumour growth. In these works, the authors consider a single category of cells whose density is denoted by $n(x,t)$ and depends on time $t\geq 0$ and position $x\in \mathbb{R}^d$. The diffusion of the density is triggered by a mechanical pressure $p=p(n)$ which is a given by a non linear function of the density $n$. Cell proliferation is modelled by a growth function $G=G(p)$ depending of the pressure. The displacement of the cells occurs with a velocity $v=v(x,t)$ related to the pressure gradient through Darcy's law. The model is written as follows,
\begin{align}
& \partial_t n + \nabla \cdot (n v) = n G(p),  \quad \mbox{ on }\mathbb{R}^+\times \mathbb{R}^d,  \label{eq:nH1} \\
& v = - \nabla p, \qquad p = p(n).  \label{eq:pH}
\end{align}
In \cite{PQV,PQTV,PV} the pressure is expressed as 
\begin{equation} \label{eq:ppH} 
p(n)= \frac{\gamma}{\gamma -1}n^{\gamma-1}.
\end{equation}  
Inserting \eqref{eq:ppH} into \eqref{eq:nH1}, \eqref{eq:pH} leads to the porous medium equation which has been widely studied \cite{Vel}. This model can be expressed as the gradient flow for the Wasserstein metric of the following energy,
\begin{equation}
 \mathcal{E}(n) = \int_{R^d}P(n(x)) ~ dx,  \label{eq:E1}
 \end{equation}
  where $P$ is a primitive of $p$, i.e. $\frac{\partial P}{\partial n}=p$. Indeed the functional derivative  $ \frac{\delta \mathcal{E}}{\delta n}$ of $\mathcal{E}$ with respect to the density $n$ acting on a density increment $\delta n$ is given by 
$$ <\frac{\delta \mathcal{E}}{\delta n},\delta n> = \lim_{h \rightarrow 0} \frac{1}{h}(\mathcal{E}(n+h\delta n) -\mathcal{E}(n)) = \int p(n)  \delta n ~ dx,$$
where $<.,.>$ is the $L^2$ duality bracket, so that we have $\frac{\delta \mathcal{E}}{\delta n}=p$.
So we can rewrite \eqref{eq:nH1} as the following equation
$$ \partial_t n + \nabla \cdot \big(n \nabla \frac{\delta \mathcal{E}}{\delta n} \big) = n G(p). $$
 The incompressible limit of this model corresponds to $\gamma \rightarrow +\infty$. It has been shown in \cite{HV} that this incompressible limit is a Hele-Shaw free boundary model which, classically, used to describe the pattern of tumor growth.
 In two dimensions, the classical Hele-Shaw problem models an incompressible viscous fluid squeezed between two parallel flat plates. As more fluid is injected, the region occupied by the fluid expands. It has been shown that the incompressible limits of many PDEs converge towards Hele-Shaw type models \cite{BBH,BI,EHKO,FriedmanHuang,GQ}. This incompressible limit of the corresponding Hele-Shaw model have been shown to be particularly relevant to tumor growth modelling \cite{PQV,PQTV}.
 
 In this paper, we present a new Gradient Flow Model (GFM) for two cells populations which is built upon the gradient flow framework presented above. We introduce a free energy $\mathcal{E}(n_1,n_2)$ that depends on the cell densities of each cell population $n_1$ and $n_2$. The free energy encompasses a term like \eqref{eq:E1} which depends on the total cell density $n = n_1+n_2$ and models both cell contact inhibition and packing. In addition we introduce active repulsion between unlike cells in order to enforce the segregation property expressed as $r=n_1n_2=0$ i.e. the two cell densities cannot be simultaneously non-zero. As this term induces repulsion forces and triggers instabilities, we also include regularising terms depending on the gradients $\nabla n_1$, $\nabla n_2$. We pursue two goals. The first one is the formal derivation of the incompressible limit model which takes the form of a two species Hele-Shaw Model (HSM). The second goal is to develop a numerical method for the GFM which enables us to illustrate the validity of the limit HSM.
 
The pressure law \eqref{eq:ppH} previously used in the literature does not prevent cells from overlapping. Indeed, with this expression, the cell density can take a value greater than $n=1$, where the value $n=1$ is supposed to be the maximal allowed cell density, corresponding to complete packing. In this paper we rather use the expression (instead of \eqref{eq:ppH}) 
 \begin{equation}
p_{\epsilon}(n)=\epsilon \frac{n}{1-n}.  \label{eq:p}
 \end{equation}
With this expression, the pressure has a singularity at $n=1$ which prevents the density to take values above $n=1$. Similar pressure laws have been used in \cite{HV}. The limit $ \gamma \rightarrow \infty$ is now replaced by $\epsilon \rightarrow 0$.

Systems with multiple populations are studied in many different areas. In chemistry, reaction-diffusion systems are used to model reacting chemical substances \cite{Dutt}. In population dynamics, these model are generalised into cross-diffusion systems in which the movement of one species can be induced by the gradient of the population of another species. In biology, Keller-Segel models \cite{KS} are used to model bacterial chemotaxis. Another classical example of cross-diffusion in biology is the Lotka-Volterra model \cite{Lotka}, which describes the dynamics of a predator-prey system. These have been extended to nonlinear diffusion Lotka-Volterra systems to model cell populations \cite{BGHP,BT}. In the context of tumor growth, systems with different types of cells have been studied (such as healthy/tumor cells, proliferative/quiescent cells \cite{BPM,BHIM,SC}). Among these models, some conserve the segregation property  \cite{BGHP,BT,GSV}, which is one of the characteristic we focus on. On the other hand, some models impose the species non mixing property. This is the case of the Cahn-Hilliard equation, which describes the process of phase separation \cite{ES}. In this model, each phase tends to regroup in one domain and can erase the other phase. In our model, the total number of cells of each species remains constant (in the absence of growth terms and up to possible boundary effects). Though, there are some similarities between our model and the Cahn-Hilliard equation. 

The paper is divided into the following five sections. In Section 2, the main results are exposed: the two population model is described; the formal incompressible limit theorem is stated; numerical simulations are shown in support and a discussion is provided. Then, Section 3 contains the derivation of the model from the free energy.  The formal proof of the convergence to the Hele Shaw free boundary problem is in Section 4. Finally, Section 5 is devoted to the description of the numerical scheme. A short conclusion is given in Section 6.

\section{Main results}

\subsection{Introduction of the gradient flow model}

In this paper, we consider two densities of cells denoted by $n_1(t,x)$ and $n_2(t,x)$ and the corresponding pressures $p_1(t,x)$ and $p_2(t,x)$ that depend on time $t\geq 0$
and position $x\in \mathbb{R}^d$. We derive our model from the single cell model exposed in the introduction and define the free energy $\mathcal{E}(n_1, n_2)$ expressed as 
follows:
\begin{equation}
 \mathcal{E}(n_1, n_2) = \int_{\mathbb{R}^d} P_{\epsilon}(n_1+n_2)dx + \int_{\mathbb{R}^d} Q_m(n_1n_2)dx + \frac{\alpha}{ 2} \int_{\mathbb{R}^d} ( | \nabla n_1|^2 + | \nabla n_2|^2)dx ,  \label{eq:E12}
 \end{equation}
 where $ \alpha$ is a diffusion parameter and $\epsilon, m$ are parameters of the pressure laws.
The first term corresponds to the pressure building up from the volume exclusion constraint and the second term is a repulsion pressure between the two different categories of cells. The last term represents cohesive energy penalising strong gradients of either cell densities.

 We assume that the two categories of cells have identical geometrical characteristics, so that the volume exclusion pressure resulting from either category of cells is similar and the total volume exclusion pressure is just a function of the total cell density. The expression of the pressure $p_{\epsilon}$ is given by \eqref{eq:p} and $\epsilon$ a parameter supposed to be small.
 The repulsion pressure is a novel aspect of the model and is defined by
\begin{equation}
q_m(r)= \frac{m}{m-1}(1+r)^{m-1},  \label{eq:q}
 \end{equation}
 with $q_m=Q'_m$.
This expression imposes segregation when $m$ is going to infinity. This can be seen thanks to the equality $ (1+r) q_m(r) = (1-\frac{1}{m})^{\frac{1}{m-1}} q ^{\frac{m}{m-1}} $. Passing to the limit $m \rightarrow \infty$, we obtain $r^{\infty}q^{\infty}=0$. Since $q^{\infty}>1$, this implies that $r^{\infty}=0$, which expresses the segregation property (the cell densities cannot be simultaneously non-zero).

For the sake of simplicity, we omit the parameters $ \epsilon, m, \alpha$ in the notations of unknown function $n_1$, $n_2$, $p_1$ and $p_2$. From the free energy \eqref{eq:E12}, in Section 3, we derive the following system of equations,
\begin{align}
& \partial_t{n_1} -  \nabla_x ({n_1} \nabla{p_1}) + \alpha \nabla_x ({n_1} \nabla ( \Delta {n_1}))= {n_1} G_1({p_1}),  \label{eq:n1} \\
&\partial_t{n_2} - \nabla_x ({n_2} \nabla {p_2}) + \alpha \nabla_x ({n_2} \nabla ( \Delta {n_2})) = {n_2} G_2({p_2}), \label{eq:n2} \\
& {p_1}=  p_{\epsilon}({n_1}+{n_2})+  {n_2} q_m({n_1}{n_2}), \label{eq:p1}\\
& {p_2}=  p_{\epsilon}({n_1}+{n_2})+  {n_1}q_m({n_1}{n_2}) . \label{eq:p2}
\end{align} 
where $G_1$ and $G_2$ are growth functions depending of the pressures $p_1$ and $p_2$, respectively.
The particular case $q_m =0$ and $\alpha = 0$ has been studied in \cite{BGHP,BT,GSV} where it has been shown that there exists a segregation property between the species, provided that the initial conditions are segregated \cite{BPM,BHIM,CFSS,KM}. The present paper treats a different case, as initially the two species may be mixed and the model drives them to a segregated state after some time, except for a thin interface depending on $\alpha$ and $m$. This is consistant with the biological observation that some mixing between the cell species occurs across the interface. The aim of this paper is to investigate the incompressible limit of this model \eqref{eq:n1}-\eqref{eq:p2}, which consists in letting $\epsilon$ and $\alpha$ going to 0 and $m$ going to infinity in the system.

\subsection{Formal limit}
In this section, we obtain convergence results of the model when $\epsilon, \alpha \rightarrow 0$ and  $m \rightarrow \infty$. First we list some assumptions.
As for the growth function, we assume:
\begin{equation}\label{hypG}
  \left\{
      \begin{aligned}
         &\exists\, G_m>0, \quad \| G_1 \|_{\infty} \leq G_m, \quad \| G_2 \|_{\infty} \leq G_m,\\
        & G_1', G_2' <0, \quad \mbox{ and }  \exists\, p^*_1, p^*_2>0, \quad G_1(p^*_1)=0 \mbox{ and } G_2(p^*_2)=0,\\
        & \exists\, \gamma>0, \quad \min(\min_{[0,p^*1]} |G_1'|,\min_{[0,p*1]} |G_2'| )= \gamma.\\
      \end{aligned}
    \right. 
\end{equation}
These assumptions stem from biological considerations. As the pressure increases in the tissue, cell division occurs less frequently, until eventually the pressure reaches a critical value, which either stops the growth or starts to trigger cell death.  
As for the initial conditions, we assume that there exists $\epsilon_0>0$ such that, for all $\epsilon\in (0,\epsilon_0)$,
\begin{equation}\label{hypini}
  \left\{
      \begin{aligned}
      & 0 \leq n_1^{\rm ini}, \quad 0 \leq n_2^{\rm ini}, \quad 0 \leq n^{\rm ini}:=n_1^{\rm ini}+n_2^{\rm ini},\\
       & p^{\rm ini}_{\epsilon}:= \epsilon \frac{n^{\rm ini}}{1-n^{\rm ini}} \leq p^* := \max(p^*_1,p^*_2). \\
             \end{aligned}
    \right.
\end{equation}
Thanks to these assumptions we can establish a priori estimates on ${n_1}$ and ${n_2}$. With a Stampaccchia method, we show the positivity of the densities. The $L^1$ bounds on ${n_1}$ and ${n_2}$ follow from this result. The initial $L^\infty$ bounds ensure that the densities stay below the critical value 1. However, in order to pass to the limit, we need more a priori estimates. In the following theorem, we only provide a formal limit as proof rigorous would require obtaining a priori estimates on the two densities, the pressure and on their derivatives in time and in space. The emergence of terms of the type $\nabla n \nabla r$ in the space derivatives of the densities makes it difficult to obtain required estimates since we are not able to control these terms. The same kind of problem is encountered for the time derivative.
 
 The main result is the stated following theorem.
 \begin{theo}\label{TH1} (Formal limit)
Let $T>0$, $Q_T=(0,T)\times \mathbb{R}^d$. 
Let $G_1$, $G_2$ and $n^{\rm ini}$, $n_1^{\rm ini}$, $n_2^{\rm sini}$ satisfy assumptions \eqref{hypG} and \eqref{hypini}. Suppose the limits of the densities ${n_1}$, ${n_2}$ , of the pressure $p_{\epsilon}$ and of $q_m$  as $\epsilon \rightarrow 0$, $m \rightarrow +\infty $ and $\alpha \rightarrow 0$ exist and are denoted by $n_1^{\infty}, n_2^{\infty} $, $p^{\infty} $ and $q^{\infty} $. If the convergence is strong enough then these limits satisfy
\begin{align}
\partial_t n_1^{\infty} - \nabla \cdot ( n_1^{\infty} \nabla p_1^{\infty}) = n_1^{\infty} G_1( p_1^{\infty}), \label{eq:n10} \\
\partial_t n_2^{\infty} - \nabla \cdot ( n_2^{\infty} \nabla  p_2^{\infty}) = n_2^{\infty} G_2( p_2^{\infty}),   \label{eq:n20} \\
\partial_t n^{\infty} - \Delta  p^{\infty} = n_1^{\infty} G_1( p_1^{\infty}) +n_2^{\infty} G_2( p_2^{\infty}),   \label{eq:n0} 
\end{align}
with $$ p_1^{\infty}= p^{\infty}+n_2^{\infty} q^{\infty} \quad \text{ and } \quad p_2^{\infty}= p^{\infty}+n_1^{\infty} q^{\infty} .$$
In addition, we have
\begin{align}\label{n0p0}
(1-n^{\infty})p^{\infty}=0,
\end{align}
with $ n^{\infty} = n_1^{\infty} +n_2^{\infty}$,
 \begin{align}\label{segre}
n_1^{\infty}n_2^{\infty}=0,
\end{align}
and the complementary relation
\begin{align} \label{compl}
 (p^{\infty})^2(\Delta p^{\infty}+ n_1^{\infty}G_1(p^{\infty})+ n_2^{\infty}G_2(p^{\infty})) =0.
\end{align}
\end{theo}

Eq. \eqref{n0p0}  suggests to decompose the domain in two parts. We consider the domain $ \Omega(t) =\{ x ~|~  {p}^{\infty} (x,t)>0 \}$ and the complementary domain, where the pressure is equal to 0. Notice that $\Omega(t) \subset \{ x ~|~ {n}^{\infty}(x,t) =1 \}$. Moreover, the two domains coincide almost everywhere. Indeed if ${n}^{\infty}=1$ and ${p}^{\infty}=0$, the density will continue to grow and the total density will become greater than the maximum packing value. Thanks to the segregation property \eqref{segre} we can decompose $\Omega(t)$ in two subdomain, $ \Omega_1(t) = \Omega(t) \cap \{ x ~|~ {n_1}^{\infty}(x,t) =1 \}$ and $ \Omega_2(t) = \Omega(t) \cap \{  x ~|~{n_2}^{\infty}(x,t) =1 \}$. Then it is verified that $ \Omega_1(t)$ and $ \Omega_2(t)$ are disjoint and that their reunion is $ \Omega(t)$. It is interesting to remark that

\begin{equation} \label{p1inf}
p_1^{\infty}= \left\{
      \begin{aligned}
       p^{\infty} \text{ in } \Omega_1(t),\\
       p^{\infty}+ q^{\infty} \text{ in } \Omega_2(t),\\
        0 \text{ outside } \Omega_1(t)\cup \Omega_2(t),
      \end{aligned}
    \right.
\end{equation}
\begin{equation} \label{p2inf}
p_2^{\infty}= \left\{
      \begin{aligned}
       p^{\infty}+ q^{\infty} \text{ in } \Omega_1(t),\\
       p^{\infty} \text{ in } \Omega_2(t),\\
        0 \text{ outside } \Omega_1(t)\cup \Omega_2(t),
      \end{aligned}
    \right.
\end{equation}
with $  q^{\infty} = 1$.

 The equation for the pressure inside  $ \Omega(t)$ is deduced from the complementary relation \eqref{compl}:
\begin{align*}
\Delta p^{\infty} + n_1^{\infty}G_1(p^{\infty})+ n_2^{\infty}G_2(p^{\infty}) = 0 \text{ on } \Omega(t).
\end{align*}
Since the limit pressure is continuous, the pressure is equal to zero on the boundary $\partial \Omega(t)$. However the normal derivative of the pressure is non-zero and controls the movement of the domain boundary. Knowing that in the system \eqref{eq:n1} - \eqref{eq:p2} the densities $n_1$ and $n_2$ diffuse with velocities $\nabla p_1$ and $\nabla p_2$ respectively and given that at the limit, according to \eqref{p1inf} and \eqref{p2inf}, these two pressures are equal to $p^{\infty}$, we can deduce that the domain boundary and the interface are moving with a normal velocity,
\begin{equation} \label{V}
 V_{\upsilon} = - \nabla p^{\infty} \cdot \upsilon  ,
 \end{equation}
 where $\upsilon$ is the outward normal  vector.
 
 \subsection{Numerical validation}
 
\subsubsection{Analytical solution of the Hele-Shaw model}

The HSM is characterised by the complementary relation \eqref{compl} and the velocity of the boundary given by  \eqref{V}. In the one-dimensional case (1D), the solution of this problem can be computed. We will consider $G(p) = g(p^*-p)$, with $p^*$ the maximum pressure and $g$ the growth rate, since it is one of the most common growth term in literature. In the case of one population, the complementary relation \eqref{compl} in 1D can be rewritten as
\begin{align}\label{CR1D}
 - p''(x) + g p = g p^* \mbox{ on }  \Omega(t).
 \end{align}
The solutions of this problem are of the form
   $$ p(x)= A e^{\sqrt g x}+B e^{-\sqrt g x}+ p^*,$$
   with constants A, B  depending on the boundary conditions. We compute the exact solution in two different cases which are going to be used for the numerical validation of the model.
   
   \paragraph{Example 1} We first consider the case where the two species have the same growth term
   \begin{align}\label{ex1:G}
   G_1(p) =G_2(p) =  g(p^*-p),
   \end{align}
    with one species surrounded by the other one. The density $n_2^{\mbox{\scriptsize ini}}$ is defined as the indicator function of $[x^{\mbox{\scriptsize ini}}_{\Gamma^-};x^{\mbox{\scriptsize ini}}_{\Gamma^+} ]$ and the density $n_1^{\mbox{\scriptsize ini}}$ is defined as the indicator function of $[x^{\mbox{\scriptsize ini}}_{\Sigma^-};x^{\mbox{\scriptsize ini}}_{\Sigma^+} ] \setminus{ [x^{\mbox{\scriptsize ini}}_{\Gamma^-};x^{\mbox{\scriptsize ini}}_{\Gamma^+} ]}$ where $x^{\mbox{\scriptsize ini}}_{\Gamma^-}, x^{\mbox{\scriptsize ini}}_{\Gamma^+}$ represents the interfaces between the two populations and $x^{\mbox{\scriptsize ini}}_{\Sigma^-},x^{\mbox{\scriptsize ini}}_{\Sigma^+}$ represent the exterior boundaries of the total density. More specifically,
   
  \begin{align}\label{ex1:n}
  n_1^{\mbox{\scriptsize ini}} (x) = \mathbbm{1}_{[x^{\mbox{\scriptsize ini}}_{\Sigma^-};x^{\mbox{\scriptsize ini}}_{\Gamma^-} ] }(x) + \mathbbm{1}_{[x^{\mbox{\scriptsize ini}}_{\Gamma^+};x^{\mbox{\scriptsize ini}}_{\Sigma^+} ] }(x) \quad \text{and} \quad n_2^{\mbox{\scriptsize ini}} (x) = \mathbbm{1}_{[x^{\mbox{\scriptsize ini}}_{\Gamma^-};x^{\mbox{\scriptsize ini}}_{\Gamma^+} ] } (x),
     \end{align}
   where $\mathbbm{1}_S$ is the indicator function of the set $S$.
   Since the two populations have the same growth function, the complementary relation can be treated as that of a single population of cells \eqref{CR1D} with the boundary conditions $p(x_{\Sigma^-})=0$ and $p(x_{\Sigma^+})=0$ at any time. A simple computation shows that 
 $$ p(x)=p^* (1-\frac{ \cosh(\sqrt g( \frac{x_{\Sigma^-}+x_{\Sigma^+}}{2}-x)}{ \cosh(\sqrt g \frac{x_{\Sigma^-}-x_{\Sigma^+}}{2}}) ,$$
 where $\cosh$ and $\sinh$ stand for the hyperbolic cosine and sine.
 The velocities of the exterior boundaries and of the interfaces can be easily computed:
 $$ V_{\Sigma^-} = - p^* \sqrt g \frac{ \sinh(\sqrt g\frac{x_{\Sigma^-}-x_{\Sigma^+}}{2})}{ \cosh(\sqrt g \frac{x_{\Sigma^-}-x_{\Sigma^+}}{2})} \quad \text{and} \quad
V_{\Sigma^+} =  p^* \sqrt g \frac{ \sinh(\sqrt g\frac{x_{\Sigma^-}-x_{\Sigma^+}}{2})}{ \cosh(\sqrt g \frac{x_{\Sigma^-}-x_{\Sigma^+}}{2})},$$
and
  $$ V_{\Gamma^-} = p^* \sqrt g \frac{ \sinh(\sqrt g( \frac{x_{\Sigma^-}+x_{\Sigma^+}}{2}-x_{\Gamma^-})}{ \cosh(\sqrt g \frac{x_{\Sigma^-}-x_{\Sigma^+}}{2})} \quad \text{and} \quad
V_{\Gamma^+} =  p^* \sqrt g \frac{ \sinh(\sqrt g( \frac{x_{\Sigma^-}+x_{\Sigma^+}}{2}-x_{\Gamma^+})}{ \cosh(\sqrt g \frac{x_{\Sigma^-}-x_{\Sigma^+}}{2})}.$$
Then $x_{\Sigma^{\pm}}$ and $x_{\Gamma^\pm}$ evolve respectively according to $ \frac{d }{dt} x_{\Sigma^{\pm}} = V_{\Sigma^\pm}$ and $ \frac{d }{dt} x_{\Gamma^{\pm}} = V_{\Gamma^\pm}$. Since initially $ x^{\mbox{\scriptsize ini}}_{\Sigma^-} <x^{\mbox{\scriptsize ini}}_{\Gamma^-} <x^{\mbox{\scriptsize ini}}_{\Gamma^+} < x^{\mbox{\scriptsize ini}}_{\Sigma^+} $ the density spreads. However $ |V_{\Gamma^-}| \leq V_{\Sigma^-} $ and $ |V_{\Gamma^+}| \leq V_{\Sigma^-} $ so the interface is moving more slowly than the exterior boundary. This means that the density $n_1$ is not only transported but also spreads. The density $n_2$ spreads and simultaneously pushes $n_1$.

  \paragraph{Example 2} We now consider two species having only one contact point, with different growth terms 
     \begin{align}\label{ex2:G}
     G_1(p)=g_1(p_1^*-p) \quad \mbox{and} \quad  G_2(p)=g_2(p_2^*-p).
   \end{align} 
 The initial densities are defined as indicator functions of $[x^{\mbox{\scriptsize ini}}_{\Sigma^-};x^{\mbox{\scriptsize ini}}_{\Gamma} ]$ and $[x^{\mbox{\scriptsize ini}}_{\Gamma};x^{\mbox{\scriptsize ini}}_{\Sigma^+} ]$ where $x^{\mbox{\scriptsize ini}}_{\Sigma^-},x^{\mbox{\scriptsize ini}}_{\Sigma^+}$ define the boundary of the total density and $x^{\mbox{\scriptsize ini}}_{\Gamma}$ defines the interface between the two densities. More specifically,
     \begin{align}\label{ex2:n}
      n_1^{\mbox{\scriptsize ini}} (x) = \mathbbm{1}_{[x^{\mbox{\scriptsize ini}}_{\Sigma^-};x^{\mbox{\scriptsize ini}}_{\Gamma} ] }(x) \quad \mbox{and} \quad n_2^{\mbox{\scriptsize ini}} (x) = \mathbbm{1}_{[x^{\mbox{\scriptsize ini}}_{\Gamma};x^{ini}_{\Sigma^+} ] } (x).
       \end{align} 
   The pressure follows the complementary relation \eqref{compl}, which can be rewritten as
     \begin{equation*}
    \left\{
      \begin{aligned} - p''(x) + g_1 p = g_1 p^*_1  & \mbox{ in }  \Omega_1= [x_{\Sigma^-};x_{\Gamma} ],\\
  - p''(x) + g_2 p = g_2 p^*_2  &\mbox{ in }  \Omega_2= [x_{\Gamma};x_{\Sigma^+} ],
    \end{aligned}
    \right.
\end{equation*}
  with the additional conditions  $p(x_{\Sigma^-})=0$,  $p(x_{\Sigma^+})=0$, and $p$ and $p'$ are continuous at $x_{\Gamma}$. After some computations we find,
  \begin{equation*}
   p(x)=
  \left\{
      \begin{aligned}
       & 2A_1 e^{\sqrt g_1 x_{\Sigma^-}} \sinh(\sqrt g_1(x-x_{\Sigma^-}))+p_1^*(1-e^{-\sqrt g_1(x-x_{\Sigma^-})})  \quad \text{in } [x_{\Sigma^-};x_{\Gamma} ],\\
        & 2A_2 e^{\sqrt g_2 x_{\Sigma^+}} \sinh(\sqrt g_2(x-x_{\Sigma^+}))+p_1^*(1-e^{-\sqrt g_2(x-x_{\Sigma^+})})  \quad \text{in } [x_{\Gamma};x_{\Sigma^+} ]  ,\\
       & 0  \quad \text{ in } (-\infty,x_{\Sigma^-}] \cup [x_{\Sigma^+},+\infty),
      \end{aligned}
    \right.
\end{equation*}
  with 
    \begin{equation*}
  \begin{array}{rrl}
 A_1 =  \frac{ e^{\sqrt g_2 x_{\Sigma^-}}}{D} (&p_1 (\sqrt{g_2}(1-e^{-\sqrt{g_1}*(x_{\Gamma}-x_{\Sigma^+})}) & \cosh(\sqrt{g_2}(x_{\Gamma}-x_{\Sigma^-}))  \\
   & -\sqrt{g_1} e^{-\sqrt{g_2} (x_{\Gamma}-x_{\Sigma^+})} & \sinh(\sqrt{g_2}(x_{\Gamma}-x_{\Sigma^-})))  \\
    & +p_2 \sqrt{g_2} (-(1-e^{-\sqrt{g_2}(x_{\Gamma}-x_{\Sigma^-})}) & \cosh(\sqrt{g_2}(x_{\Gamma}-x_{\Sigma^-}))   \\
   & +e^{-\sqrt{g_2}(x_{\Gamma}-x_{\Sigma^-})} & \sinh(\sqrt{g_2}(x_{\Gamma}-x_{\Sigma^-})))) ,
\end{array}
\end{equation*}
and
    \begin{equation*}
  \begin{array}{rrl}
 A_2 =  \frac{ e^{\sqrt g_1 x_{\Sigma^+}}}{D}(&p_1\sqrt{g_1}((1-e^{-\sqrt{g_1}(x_{\Gamma}-x_{\Sigma^-})})  & \cosh(\sqrt{g_1}(x_{\Gamma}-x_{\Sigma^-})) \\
  &- e^{-\sqrt{g_1} (x_{\Gamma}-x_{\Sigma^-})} & \sinh(\sqrt{g_1}(x_{\Gamma}-x_{\Sigma^-})))  \\
  &+ p_2( \sqrt{g_2}(1-e^{-\sqrt{g_2}(x_{\Gamma}-x_{\Sigma^+})}) & \cosh(\sqrt{g_2}(x_{\Gamma}-x_{\Sigma^+}))   \\
  & -  \sqrt{g_1} e^{-\sqrt{g_2}(x_{\Gamma}-x_{\Sigma^+})} & \sinh(\sqrt{g_1}(x_{\Gamma}-x_{\Sigma^-})))).
  \end{array}
\end{equation*}
From this expression we can compute the speed of the exterior boundary and of the interface by computing the derivative of the pressure and evaluating it at $x^{\mbox{\scriptsize ini}}_{\Sigma^-},x^{\mbox{\scriptsize ini}}_{\Sigma^+}$ and $x^{\mbox{\scriptsize ini}}_{\Gamma}$.

\subsubsection{Numerical validation}

In this section we use the numerical scheme presented in Section 4 to illustrate that the solution of the GFM converges to the corresponding solution of the HSM described in the previous section 2.3.1. In order to facilitate the comparison with the analytical solutions of the HSM shown in the previous section, the simulations are performed in one dimension. To compare the two models we initialize them with the same initial configuration: the densities are taken as segregated indicator functions. We consider the two different cases with different growth functions and different initial conditions, corresponding to Examples 1 and 2 of the previous section.

The parameters are $\epsilon =0.1$, $m =10$, $\alpha =0.01$. Another parameter which is introduced for the Relaxation Gradient Flow Model (RGFM) in Section 4 is given the value $\nu =0.001$. We plot on the same figure the numerical solutions of the GFM (solid line) and the HSM (dashed line). The species $n_1$ and $n_2$ are plotted in red and blue, respectively.

The first case defined as Example 1 in Section 2.3 is illustrated in Fig. \ref{fig1}. We use the growth function \eqref{ex1:G} with the parameters $g=1$ and $p^*=10$. The boundary and interface are taken as $x^{\mbox{\scriptsize ini}}_{\Sigma^\pm} = \pm 1.4$ and $x^{\mbox{\scriptsize ini}}_{\Gamma^\pm} = \pm 0.6$. Then the initial density of the HSM are defined by the formula \eqref{ex1:n}. As initial conditions of the GMF, we take 
 \begin{align}\label{ex1:nGFM}
 n_1^{ini} (x)= 0.98 ~ ( \mathbbm{1}_{[-1.4;-0.6 ] }(x)+\mathbbm{1}_{[0.6;1.4 ] }(x)) \quad \mbox{and} \quad n_2^{ini}(x) = 0.98 ~ \mathbbm{1}_{[-0.6;0.6 ] } (x).
 \end{align} 
We take  $n_1=n_2=0.98$ initially as it is close to the singular value 1. In this example the solution of the GFM is expected to be close to the particular solution of the HSM given in Example 1.
   
   \begin{figure}[!h]
   \centering
\includegraphics[width=1\linewidth]{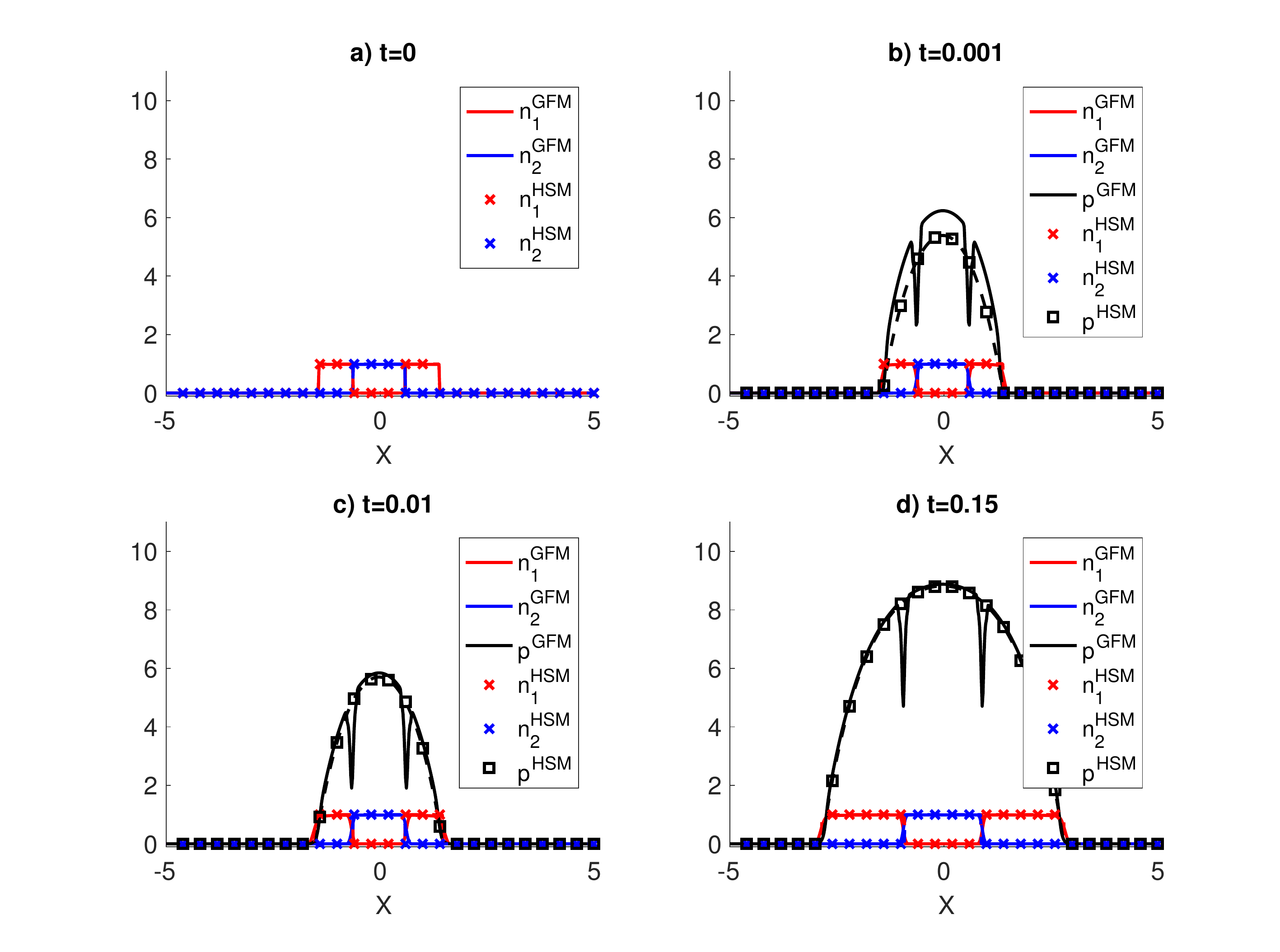}
\caption{Snapshots of densities $n_1$ (red), $n_2$ (blue) and pressure $p$ (black) for GFM (solid line) and HSM (dashed line with marker) from Example 1 at different times: (a) $t=0$ , (b)  $t=0.001$, (c) $t=0.01$, (d) $t=0.15$. Initial conditions: density of the GFM from Eq.~\eqref{ex1:nGFM}, density of the HSM from Eq.~\eqref{ex1:n}, growth functions from Eq.~\eqref{ex1:G}.}
\label{fig1}
\end{figure}
   
   The second case is illustrated in Fig. \ref{fig2}. The growth functions are defined by \eqref{ex2:G} with the parameters $g_1=g_2=1$, $p_1^*=20$ and $p_2^*=10$.  As initial conditions of the GMF we take 
    \begin{align}\label{ex2:nGFM}
    n_1^{ini} (x)= 0.98 ~ \mathbbm{1}_{[x^{\mbox{\scriptsize ini}}_{\Sigma^-};x^{\mbox{\scriptsize ini}}_{\Gamma} ] }(x) \quad \mbox{and} \quad  n_2^{ini}(x) = 0.98 ~ \mathbbm{1}_{[x^{\mbox{\scriptsize ini}}_{\Gamma};x^{\mbox{\scriptsize ini}}_{\Sigma^+} ] } (x),
     \end{align} 
    where $x^{\mbox{\scriptsize ini}}_{\Sigma^\pm} = \pm 1.4$ and $x^{\mbox{\scriptsize ini}}_{\Gamma} =  0$. The initial density of the HSM are defined with formula \eqref{ex2:n}. In this example the solution of the GFM is expected to be close to the  the particular solution of the HSM given in Example 2.    
   \begin{figure}[!h]
   \centering
\includegraphics[width=1\linewidth]{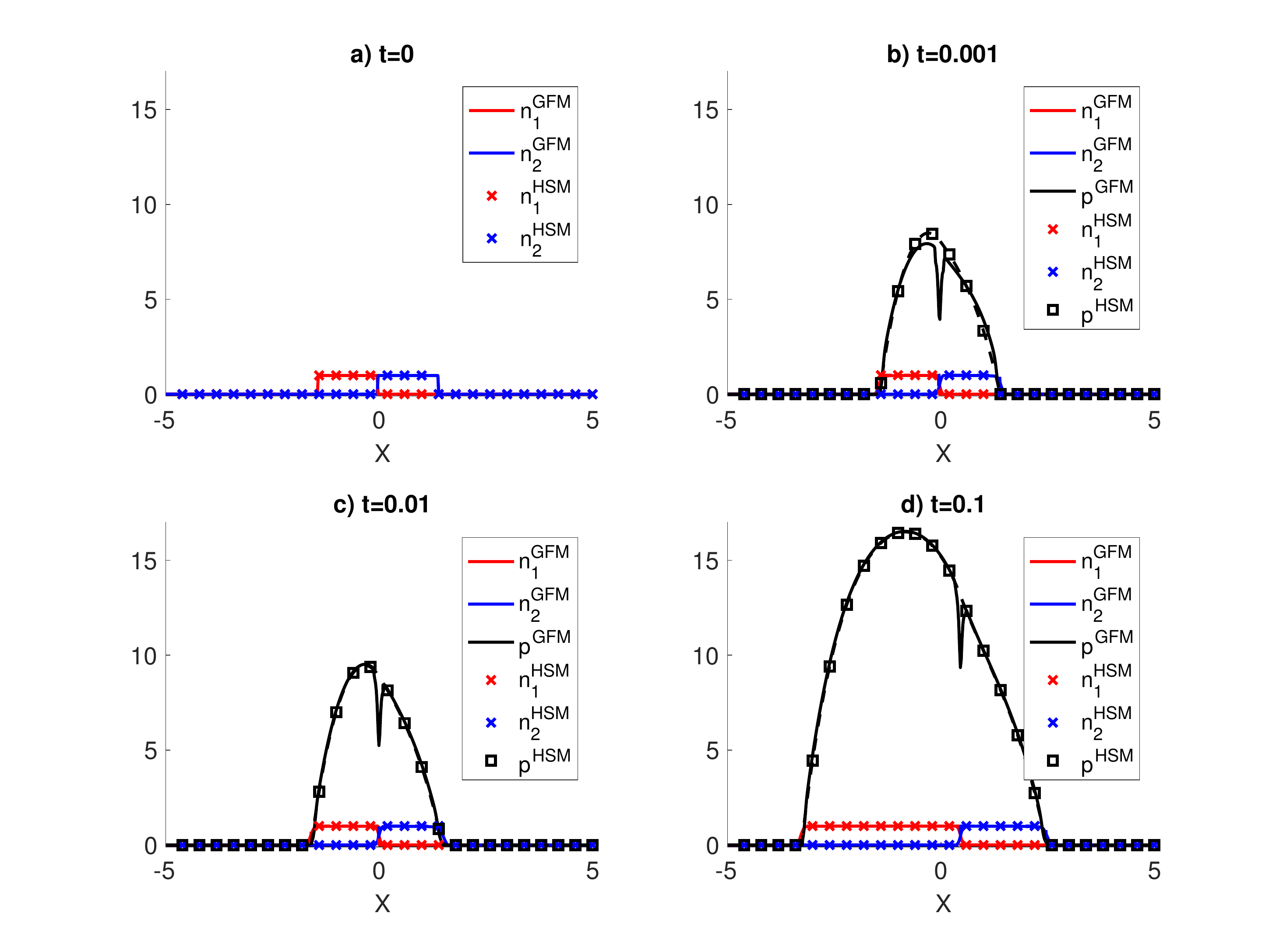}
\caption{Snapshots of densities $n_1$ (red), $n_2$ (blue) and pressure $p$ (black) for GFM (solid line) and HSM (dashed line with marker) from Example 2 at different times: (a) $t=0$ , (b)  $t=0.001$, (c) $t=0.01$, (d) $t=0.15$. Initial conditions: density of the GFM from Eq.~\eqref{ex2:nGFM}, density of the HSM from Eq.~\eqref{ex2:n}, growth functions from Eq.~\eqref{ex2:G}.
 }
\label{fig2}
\end{figure}

We can draw the same conclusions about the approximation of the HSM by the GFM in the two sets of simulations displayed in Figs.  \ref{fig1} and  \ref{fig2} (where panel (a) is for the initial configuration and panels (b), (c), (d) are for times $t=0.001$, $t=0.01$, $t=0.15$ or $t=0.25$ respectively). First of all, we notice that the approximation is excellent except for the pressure at small times, and at the interface between the two cell population. While the pressure of the HSM is smooth wherever the total density is positive, the pressure of the GFM exhibits sharp ditches at the interfaces between the two populations. Otherwise the pressures given by the two models are similar except at small times. In Fig.~\ref{fig1}(b), the pressure of the GFM is larger than that of the HSM, whereas on Fig.~\ref{fig2}(b) it is the opposite. However, we do not observe any impact of this difference on the densities. On Figs.~\ref{fig2}(c) and \ref{fig2}(d), the two pressures are almost the same for the two models. Overall, the dynamics of the GFM and the HSM are quite similar. In Fig.~\ref{fig1}, we observe that the blue (inner) species pushes the red (outer) species in order to be allowed to grow. In Fig.~\ref{fig2}, the red species (leftmost) which has the biggest growth rate pushes the blue species (rightmost) toward the right. Therefore by growing more rapidly, the red species exert a bigger pressure on the blue species than the blue species exerts on the red species. This pressure imbalance which is visible on Fig.~\ref{fig2}, triggers the motion of the interface towards regions of lower pressure.

\subsubsection{Discussion}

The results of the simulations exposed in Figs.~\ref{fig1} and \ref{fig2} show that the dynamics of the two models are almost identical after some time. However, at the beginning of the simulation we observe some differences between the pressure of the GFM and the HSM. We notice that initially in the GFM, the densities are taken as indicator functions, so at the start the pressure $p_{\epsilon}$ of the HSM is also an indicator function whereas that of the HSM is continuous. It takes some time for the scheme to smooth it out, and then the two pressures coincide. In addition, the initial value of the GFM in these simulations is fixed at $98\%$ of the packing density 1.  The growth term produces an increase of the densities to their upper bound value. The densities equate then to $n_M= p_{\epsilon}^{-1}(p^*)= \frac{p^*}{\epsilon +p^*}$ (because the pressure stays below $p^* = \max(p_1^*,p_2^*)$, see also \eqref{hypini}). During this period the GFM and the HSM are not synchronized. After a short transient, we observe that the pressures of the two models have the same shape.

 The major difference between the two models is at the interfaces between the two populations. As result of the fourth order term, the densities of the GFM are not perfectly segregated and overlap only in a narrow spatial interval. When the parameter $\alpha$ tends to 0, the system amplifies the effect of segregation. However in the numerical simulation, the parameter $\alpha = 0.01$ allows some mixing of the two populations. This mixing is observed at the interface and creates some small disturbances on the total density (i.e. the total density is not constantly equal to $n_M$), which are then amplified in the pressure. Despite this, the speeds of the interfaces and of the boundaries seem to be close in the two models. These are quite remarkable results, considering that the parameters are not yet in the asymptotic ranges where the two models should be identical. 
 
However, it is important to remark that the simulations have only been made in the case of initially segregated populations. This was necessary to initialize the free boundary model. In the case of mixed population, simulations of the GFM have been made and can be found in Section 6. They confirm the convergence of the GFM towards a free boundary model since we observe that the system self-organizes into separated domains containing the different populations. However we can't study the convergence towards the HSM since we do not know beforehand which initial condition for the HSM will correspond to the the converged GFM.

\section{Derivation of the gradient flow model}

Eqs. \eqref{eq:n1} and \eqref{eq:n2} can be derived from the gradient flow associated with the free energy  \eqref{eq:E12}. Indeed the functional derivatives $\frac{\delta \mathcal{E}}{\delta n_1}$ and $\frac{\delta \mathcal{E}}{\delta n_2}$ of $\mathcal{E}$ with respect to $n_1$ and $n_2$, acting on a density increment $\delta n_1(x)$ and $\delta n_2(x)$ are given by

\begin{align}
 \frac{\delta \mathcal{E}}{\delta n_1} =  p_{\epsilon}(n_1+n_2)+  n_2 q_m(n_1n_2) - \alpha \Delta n_1 = {p_1}_{\epsilon,m}(n_1,n_2) - \alpha \Delta n_1 , \label{eq:dedn1} \\
 \frac{\delta \mathcal{E}}{\delta n_2} = p_{\epsilon}(n_1+n_2)+  n_1 q_m(n_1n_2) - \alpha \Delta n_2 = {p_2}_{\epsilon,m}(n_1,n_2) - \alpha \Delta n_2. \label{eq:dedn2}
 \end{align}
Indeed, the computation for the first derivative is given by
\begin{eqnarray*} 
 <\frac{\delta \mathcal{E}}{\delta n_1}, \delta n_1>&:=& \int_{\mathbb{R}^d} \frac{\delta \mathcal{E}}{\delta n_1} (n_1,n_2) ~ \delta n_1 dx \\ 
 &=& \lim_{h\rightarrow0} \frac{1}{h} \Big(\mathcal{E}(n_1+h ~ \delta n_1,n_2) - \mathcal{E}(n_1,n_2) \Big)  \\
& =&  \lim_{h\rightarrow0} \frac{1}{h}  \int_{\mathbb{R}^d} \Big( P_{\epsilon}(n_1+ h ~ \delta n_1+n_2)(x)- P_{\epsilon}(n_1+n_2)(x) \Big) ~dx  \\
 & &+  \lim_{h\rightarrow0} \frac{1}{h} \int_{\mathbb{R}^d} \Big(Q_m((n_1+h ~ \delta n_1)n_2)(x)-Q_m(n_1n_2)(x) \Big)~dx  \\
 & &+  \lim_{h\rightarrow0} \frac{\alpha}{2h} \int_{\mathbb{R}^d} \Big( |\nabla(n_1 +h \delta n_1)|^2 - |\nabla n_1 |^2 \Big) ~dx  \\
& =&\int_{\mathbb{R}^d} \Big( \frac{\partial P_{\epsilon}}{\partial n_1}(n_1+n_2)+ \frac{\partial Q_m}{\partial n_1}(n_1n_2) + \alpha \nabla n_1 \nabla \delta n_1 \Big) ~  \delta n_1 dx \\
& =&\int_{\mathbb{R}^d} \Big( p_{\epsilon}(n_1+n_2)+n_2 ~  q_m(n_1n_2) - \alpha \Delta n_1 \Big) ~ \delta n_1 dx.
\end{eqnarray*}
The expression of $\frac{\delta \mathcal{E}}{\delta n_2}$ follows from a similar computation. Therefore the gradient flow associated to the free energy \eqref{eq:E12} according to the Wasserstein metric can be written as follows \cite{Otto},
\begin{align}
 \partial_t{n_1} - \nabla \cdot(n_1 \nabla \frac{\delta \mathcal{E}}{\delta n_1}(n_1,n_2))= 0,  \label{eq:den1} \\
 \partial_t{n_2} - \nabla \cdot(n_2 \nabla \frac{\delta \mathcal{E}}{\delta n_2}(n_1,n_2))= 0.  \label{eq:den2} 
 \end{align}
The metric used here is not the traditional distance on $L^2(\mathbb{R}^d)$, but the Wasserstein distance. This distance is defined on the set of probability distributions on $\mathbb{R}^d$. It is used in a wide variety of physically meaningful equations like the porous medium equation or degenerate parabolic equations \cite{ KW, JKO,Otto}. We recover Eqs.  \eqref{eq:n1} and \eqref{eq:n2} of the GFM by adding growth terms $G_1$ and $G_2$ to \eqref{eq:den1} and \eqref{eq:den2}.

 It is worth noting that the free energy decreases along the trajectories of the equation in the absence of growth terms. Indeed, using the Green formula and the Eqs. \eqref{eq:den1} and \eqref{eq:den2}, we have
 \begin{equation} \label{eq:dte}
  \begin{aligned}
  \partial_t  \mathcal{E} (n_1,n_2)  =& \int_{\mathbb{R}^d} (\frac{\delta \mathcal{E} }{\delta n_1}(n_1(x,t),n_2(x,t)) ~ \partial_t  n_1(x,t)  \\
  & +\frac{\delta \mathcal{E} }{\delta n_2}(n_1(x,t),n_2(x,t)) ~ \partial_t  n_2(x,t) ) ~dx \\
  = & \int_{\mathbb{R}^d} ( \frac{\delta \mathcal{E} }{\delta n_1}(n_1(x,t),n_2(x,t))  \nabla \cdot (n_1(x,t) \nabla \frac{\delta \mathcal{E} }{\delta n_1}(n_1(x,t),n_2(x,t))) \\
  & +\frac{\delta \mathcal{E} }{\delta n_2}(n_1(x,t),n_2(x,t)) \nabla \cdot (n_2(x,t) \nabla \frac{\delta \mathcal{E} }{\delta n_2}(n_1(x,t),n_2(x,t)))) ~ dx  \\
   = & - \int_{\mathbb{R}^d}  n_1(x,t) |\nabla \frac{\delta \mathcal{E} }{\delta n_1}(n_1(x,t),n_2(x,t))|^2 ~ dx \\
    &- \int_{\mathbb{R}^d} n_2(x,t) |\nabla \frac{\delta \mathcal{E} }{\delta n_2}(n_1(x,t),n_2(x,t))|^2 ~ dx\leq 0.
   \end{aligned}
  \end{equation}
  Therefore, when the growth rates are set to 0, the model evolves in a way that minimize the free energy. This free energy depends of the pressures, representing the levels of volume exclusion and the segregation. Therefore the model seeks to minimize the pressures, which means moving away from situations where the volume exclusion or the segregation constraints are violated.
  
  \section{Formal proof of Theorem 1}
  
  This section is dedicated to the formal proof of Theorem 1. First, we prove Eqs. \eqref{n0p0} and \eqref{segre} that lead to the definition of the evolving domains of the HSM. Secondly, we compute the equations for the densities $n_1$, $n_2$ and $n$ in the limit $\epsilon \rightarrow 0, m  \rightarrow +\infty, \alpha  \rightarrow 0$. Third, we prove Eq. \eqref{compl}, also named complementary relation, which gives the evolution of the pressure inside the domain for the HSM. Finally, we compute the speeds of the boundaries of the domains.

Eqs. \eqref{n0p0} and \eqref{segre} follow from the pressure laws \eqref{eq:p} and \eqref{eq:q}. Indeed, $ (1-n)p_{\epsilon} = \epsilon n$, it follows in the limit $\epsilon \rightarrow 0$ that 
$$ (1-n^{\infty})p^{\infty}=0.$$
Working out the repulsion pressure law \eqref{eq:q} leads to the formula 
$$ (1+r) q_{m} = (1-\frac{1}{m})^{\frac{1}{m-1}} q_{m} ^{\frac{m}{m-1}} .$$
Taking the limit $m  \rightarrow +\infty$, we recover the segregation property, i.e.
$$ (1+r^{\infty}) q^{\infty}= q^{\infty} .$$
Since $ q_m >\frac{m}{m-1}$ for all $m$, we deduce $q^{\infty}> 1$ and 
$$r^{\infty}=n_1^{\infty}n_2^{\infty}=0.$$

Eqs. \eqref{eq:n10} and \eqref{eq:n20} follow directly from taking the limit $\epsilon \rightarrow 0, m  \rightarrow +\infty, \alpha  \rightarrow 0$ in \eqref{eq:n1} and \eqref{eq:n2}. To recover \eqref{eq:n0}, it is first interesting to remark that the solution of \eqref{eq:n1}-\eqref{eq:p2} satisfies

\begin{equation}  \label{eq:nH}
\begin{split}
\partial_t n - \Delta (H(n,r))  &+\alpha \nabla \cdot ( {n_1} \nabla( \Delta  {n_1})  + {n_2} \nabla( \Delta  {n_2}))  \\
& = {n_1} G_1({p_1}) + {n_2} G_2({p_2}),
\end{split}
\end{equation}
where
\begin{equation}  \label{eq:H}
 H(n,r) =(p_{\epsilon}(n)-\epsilon \ln (p_{\epsilon}(n)+\epsilon) + \epsilon \ln \epsilon) +( rq_m(r) +  (r+1)^{m} - q_m(r)).
 \end{equation}
Indeed, by adding \eqref{eq:n1} and \eqref{eq:n2}, we obtain
 \begin{equation} \label{eq:a}
\begin{split}
\partial_t n - \nabla \cdot (n \nabla p_{\epsilon}(n) + 2 r \nabla q_m(r) + q_m(r) \nabla r)  &+\alpha \nabla \cdot ( {n_1} \nabla( \Delta  {n_1})  + {n_2} \nabla( \Delta  {n_2}))  \\
& = {n_1} G_1({p_1}) + {n_2} G_2({p_2}).
\end{split}
\end{equation}
Then given formula \eqref{eq:p},
\begin{equation} \label{eq:b}
 \begin{aligned} 
  n \nabla p_{\epsilon}(n) &=  n ~ p'_{\epsilon}(n) \nabla n\\
  &= \epsilon \frac{n}{(1-n)^2} \nabla n\\
  &= \epsilon \frac{1-(1-n)}{(1-n)^2} \nabla n\\
  &= (p_{\epsilon}'(n) - \epsilon \frac{1}{(1-n)}) \nabla n \\
  & = \nabla (p_{\epsilon}(n)-\epsilon \ln (p_{\epsilon}(n)+\epsilon) + \epsilon \ln \epsilon),
  \end{aligned}
  \end{equation}
and given formula \eqref{eq:q},
\begin{equation} \label{eq:c}
 \begin{aligned} 
  2 r \nabla q_m(r) + q_m(r) \nabla r &=  r \nabla q_m(r) + \nabla (rq_m(r))\\
&= (1+r)\nabla q_m(r) - \nabla q_m(r) + \nabla (rq_m(r)) \\ 
& = m(1+r)^{m-1} \nabla r + \nabla ((r-1)q_m(r)) \\
& = \nabla (rq_m(r) +  (r+1)^{m} - q_m(r)), 
  \end{aligned}
  \end{equation}
  and inserting \eqref{eq:b}, \eqref{eq:c} into \eqref{eq:a} gives \eqref{eq:H}.

Since $ (r+1)^{m} = \frac{m-1}{m} (1+r) q_m(r) $ and $r^{\infty}=0$, passing to the limit as $ \epsilon \rightarrow 0, m\rightarrow\infty$, we obtain
 \begin{equation}  \label{eq:H0}
 \begin{aligned}
 H(n^{\infty},r^{\infty}) & = p^{\infty}  + r^{\infty}q^{\infty}+q^{\infty} -q^{\infty} \\
 & = p^{\infty}.
 \end{aligned}
  \end{equation}
  So, taking also the limit $\alpha \rightarrow 0$, we have:
  \begin{equation*} 
\partial_t n^{\infty} - \Delta p^{\infty} = n_1^{\infty} G_1(p_1^{\infty}) + n_2^{\infty} G_2(p_2^{\infty}).
\end{equation*}

To recover the complementary relation \eqref{compl}, we need to compute an evolution equation for the pressure. To do so, first we pass to the limit $ \alpha \rightarrow 0$ in \eqref{eq:nH} but keep the notations of ${n_1}$, ${n_2} $, ${p_1}$, ${p_2} $, $p_{\epsilon} $, $q_{m} $ for the limit. We multiply the resulting equation by $p_{\epsilon}'(n)$ and obtain,
\begin{equation} \label{eq:palp}
\partial_t p_{\epsilon} - p'_{\epsilon}(n)\Delta (H(n,r))= p'_{\epsilon}(n)( {n_1} G_1({p_1}) + {n_2} G_2({p_2})).
\end{equation}
Multiplying the Eq. \eqref{eq:palp} by $\epsilon$ and taking into account that $p_{\epsilon}'(n)= \frac{1}{\epsilon} (p_{\epsilon}+\epsilon)^2$ yields
\begin{equation*}
\epsilon \partial_t p_{\epsilon} -   (p_{\epsilon}+\epsilon)^2\Delta (H(n,r))  =(p_{\epsilon}+\epsilon)^2( {n_1} G_1({p_1}) + {n_2} G_2({p_2})).
\end{equation*}
We now take the limit $\epsilon \rightarrow 0$, and denoting the pressure $p_{\epsilon} $ at the limit by $p^{\infty} $ but keep denoting the limit of the densities and the repulsion pressure by ${n_1}$, ${n_2} $, ${p_1}$, ${p_2} $ ,$q_{m} $, we get
\begin{equation*}
- {p^{\infty}}^2 \Delta (H(n^{\infty},r^{\infty})) = {p^{\infty}}^2( {n_1}G_1({p_1}) + {n_2}G_2({p_2})). 
\end{equation*}
Then passing to the limit $ m\rightarrow \infty$, and using the expression \eqref{eq:H0} for the limit of H, we recover the complementary relation \eqref{compl}.
 This concludes the formal proof of the theorem.
 
In order to recover the speed of the boundary of the Hele-Shaw model, we first focus on the speed of the exterior boundary. Thanks to Eq. \eqref{eq:n0}, for all $\varphi \in C_c^{\infty}(\mathbb{R}^d)$,
 \begin{equation*}
\int_{\mathbb{R^d}} \partial_t n^{\infty} ~ \varphi = \int_{\mathbb{R^d}}  p^{\infty} \Delta  \varphi  +  \int_{\mathbb{R^d}} n_1^{\infty} ~ G_1( p_1^{\infty}) ~ \varphi + \int_{\mathbb{R^d}} n_2^{\infty} ~ G_2( p_2^{\infty}) ~ \varphi .
\end{equation*}
Therefore applying the Green formula twice, yields
 \begin{equation*}
 \begin{aligned}
 \partial_t  \int_{\Omega(t)} n^{\infty} ~ \varphi  =  & \int_{\Omega_1(t)}  p^{\infty} \Delta  \varphi+ \int_{\Omega_2(t)}  p^{\infty} \Delta  \varphi +  \int_{\Omega_1(t)} G_1( p_1^{\infty}) ~ \varphi + \int_{\Omega_2(t)}  G_2( p_2^{\infty}) ~ \varphi \\
 =  &  \int_{\Omega_1(t)} (\Delta p^{\infty} + G_1( p_1^{\infty}))  \varphi - \int_{\partial \Omega_1(t)} \frac{ \partial p^{\infty}}{\partial \upsilon} \varphi   + \int_{\partial \Omega_1(t)}  p^{\infty}   \frac{ \partial \varphi}{\partial \upsilon}   \\
 & + \int_{\Omega_2(t)} (\Delta p^{\infty} + G_2( p_2^{\infty}))  \varphi - \int_{\partial \Omega_2(t)} \frac{ \partial p^{\infty}}{\partial \upsilon} \varphi  + \int_{\partial \Omega_2(t)}  p^{\infty}   \frac{ \partial \varphi}{\partial \upsilon}   \\
  = & \int_{\Omega(t)} (\Delta p^{\infty} + n_1^{\infty} G_1( p_1^{\infty})+ n_2^{\infty} G_2( p_2^{\infty})) ~ \varphi \\
 & - \int_{\partial \Omega(t)} \frac{ \partial p^{\infty}}{\partial \upsilon} \varphi  -  \int_{\Gamma(t)} [\frac{ \partial p^{\infty}}{\partial \upsilon}]_{12} ~\varphi + \int_{\partial \Omega(t)}  p^{\infty}   \frac{ \partial \varphi}{\partial \upsilon}  + \int_{\Gamma(t)}  [p^{\infty}]_{12}   \frac{ \partial \varphi}{\partial \upsilon} , \\
 \end{aligned}
 \end{equation*}
 where $\Gamma(t) = \Omega_1(t) \cap \Omega_2(t)$ is the interface between the two domains $\Omega_1(t)$ and $\Omega_2(t)$. On $\partial \Omega(t)$, $\partial \Omega_1(t)$, $\partial \Omega_2(t)$ the unit normals $\upsilon$ are directed outwards to the respective domains. On $\Gamma(t)$, the normal $\upsilon$ is directed from $\Omega_2(t)$ towards $\Omega_1(t)$. For $x \in \Gamma$, we denote
 $$ [\frac{ \partial p^{\infty}}{\partial \upsilon}]_{12}  = \big (\nabla ( \left. p^{\infty} \right \rvert_{\Omega_2})(x)- \nabla (\left. p^{\infty}\right \rvert_{\Omega_1})(x) \big)\cdot \upsilon(x),$$
 and
 $$ [p^{\infty}]_{12}  = ( \left. p^{\infty} \right \rvert_{\Omega_2})(x)- (\left. p^{\infty}\right \rvert_{\Omega_1})(x) .$$
 From Eq. \eqref{compl}, the first integral is equal to 0. In addition, the pressure is equal to zero on the boundary $\partial \Omega(t)$ so that
  \begin{equation} \label{eq:dtnphi}
 \begin{aligned}
  \partial_t  \int_{\Omega(t)} n^{\infty} ~ \varphi  =  & - \int_{\partial \Omega(t)} \frac{ \partial p^{\infty}}{\partial \upsilon} \varphi  -  \int_{\Gamma(t)} [\frac{ \partial p^{\infty}}{\partial \upsilon}]_{12} ~\varphi + \int_{\Gamma(t)}  [p^{\infty}]_{12}   \frac{ \partial \varphi}{\partial \upsilon} . 
  \end{aligned}
 \end{equation}
Moreover since $ n^{\infty} = 1$ in the domain $\Omega(t)$, we have
   \begin{equation} \label{eq:dtnphi2}
   \partial_t  (\int_{\Omega(t)} n^{\infty} ~ \varphi ~ dx)  =  \int_{\partial \Omega(t)} V_{\partial\Omega(t)} ~ n^{\infty} ~ \varphi ,
    \end{equation}
where $V_{\partial\Omega(t)}$ is the normal speed of the boundary $\partial\Omega(t)$ in the outwards direction.
Since Eqs. \eqref{eq:dtnphi} and  \eqref{eq:dtnphi2} are verified for all $\varphi \in C_c^{\infty}(\mathbb{R}^d)$,  we deduce that 
\begin{equation} 
   V_{\partial\Omega(t)} = - \frac{ \partial p^{\infty}}{\partial \upsilon} ,
    \end{equation}
    and at the interface $\Gamma$,
    \begin{equation} 
   [p^{\infty} ]_{12}=0 \quad \mbox{and} \quad [\frac{ \partial p^{\infty}}{\partial \upsilon}]_{12}=0.
    \end{equation}
    Then, both $p^{\infty}$ and its normal derivative are continuous at the interface $\Gamma$. To find the velocity at the interface the same method needs to be applied on either $n_1^{\infty}$ or $n_2^{\infty}$ and it leads to
    \begin{equation} 
   V_{\Gamma(t)} = - \frac{ \partial p^{\infty}}{\partial \upsilon} , 
    \end{equation}
where $V_{\Gamma(t)}$ is the speed of the boundary $\Gamma(t)$ measured in the direction of $\upsilon$. This gives \eqref{V}.

If we consider the case where the initial densities of the HSM are defined by $n_1^{\infty}= \mathbbm{1}_{B_{R_1}(t)}$ and $n_2^{\infty}(t,x)= \mathbbm{1}_{B_{R_2}(t) \backslash B_{R_1}(t)}$ with $B_R$ the ball of center 0 and radius $R$ and $R_1(t)<R_2(t)$, then $n^{\infty} = \mathbbm{1}_{B_{R_2}(t)}$, so
 $$ \partial_t n^{\infty} = R_2'(t) \delta_{\{|x|=R_2(t)\}}.$$
Using that  the pressure and its normal derivative are continuous through the interface Eq. \eqref{eq:dtnphi}, and using radial symmetry, we obtain 
 $$  \partial_t n^{\infty} = \Delta p^{\infty} + n_1^{\infty}G_1(p^{\infty}) + n_2^{\infty}G_2(p^{\infty}) = |\nabla p^{\infty}| \delta_{\{|x|=R_2(t)\}}.$$ 
 Then, $$ R_2'(t) = |\nabla p^{\infty}(R_2(t),t)|,$$ where, by abuse of notation, we have denoted $p^{\infty}(r,t)$ the constant value of $p^{\infty}(.,t)$ on the surface $\{ |x|=r\}$.
 Similarly, in applying the same result on the species $n_1^{\infty}$, we obtain
 $$ R_1'(t) = |\nabla p^{\infty}(R_1(t),t)|.$$

\section{Numerical method}

This section is devoted to the derivation of a numerical scheme used to perform the simulations of the GFM presented in Section 1.3. Since the equations for the densities are of fourth order, we first introduce a relaxation model, in which the fourth order terms are replaced by a coupled system of second order equations. Then, we present the numerical scheme together with some of its properties and CFL stability condition. The last part contains some simulations without initial segregation.

\subsection{Relaxed model}
In order to simplify the computation of system \eqref{eq:n1} - \eqref{eq:p2}, we lower its order through a suitable relaxation approximation. The relaxation system, depending on the relaxation parameter $\nu$, is written as follows
\begin{align}
& \partial_t n_1^{\nu} -\nabla \cdot (n_1^{\nu}\nabla p_1^{\nu})  + \frac{1}{\nu} \nabla \cdot \big( n_1^{\nu} \nabla (c_1^{\nu}-n_1^{\nu})\big) = n_1^{\nu} G_1(p_1^{\nu}), \label{eq:n1r} \\
&\partial_t n_2^{\nu} - \nabla \cdot (n_2^{\nu}\nabla p_2^{\nu})  +\frac{1}{\nu} \nabla \cdot \big(n_2^{\nu} \nabla (c_2^{\nu}-n_2^{\nu})\big)= n_2^{\nu} G_2(p_2^{\nu}),\label{eq:n2r} \\
& \alpha \Delta c_1^{\nu} = \frac{1}{\nu} (c_1^{\nu}-n_1^{\nu}) ,\label{eq:c1r} \\
& \alpha \Delta c_2^{\nu} = \frac{1}{\nu} (c_2^{\nu}-n_2^{\nu}) . \label{eq:c2r} 
\end{align}
where $c_1$ and $c_2$ are two new variables. The fourth order terms in \eqref{eq:n1} - \eqref{eq:p2} are now replaced by second order terms supplemented with Poisson equations which define $c_1$ and $c_2$. This type of relaxation has been used for examples, for the numerical approximation of the Navier-Stokes-Korteweg equations \cite{CDN, CR}. The Relaxation Gradient Flow Model (RGFM) formally converges toward the GFM  when $\nu \rightarrow 0$. Here we give supporting elements for this statement. We suppose that $(n_1^{\nu})$ and $(n_2^{\nu})$ converge when $\nu$ goes to 0. Since the two species play the same role, we only consider $n_1$. Inserting \eqref{eq:c1r} in \eqref{eq:n1r} we have,
$$ \partial_t n_1^{\nu} -\nabla \cdot (n_1^{\nu}\nabla p_1^{\nu})  + \alpha \nabla  \cdot \big(n_1^{\nu} \nabla \Delta c_1^{\nu} \big)=n_1^{\nu} G_1(p_1^{\nu}).$$
Then to prove the convergence we need to show that $c_1^{\nu} \rightarrow n_1$  when $\nu \rightarrow 0$. We can rewrite \eqref{eq:c1r} as
$$ -\alpha \Delta (c_1^{\nu}-n_1^{\nu}) + \frac{1}{\nu} (c_1^{\nu}-n_1^{\nu}) =\alpha \Delta n_1^{\nu} .$$
So that
$$  c_1^{\nu}-n_1^{\nu} =\alpha (-\alpha \Delta+ \frac{1}{\nu})^{-1} \Delta n_1^{\nu} .$$
Using Fourier transformation, we obtain:
$$  \hat{c}_1^{\nu}-\hat{n}_1^{\nu} =- \frac{ \alpha |\xi|^2}{\alpha |\xi|^2+ \frac{1}{\nu}} \hat{n }_1^{\nu} \quad \forall \xi.$$
And then,
$$ \lim_{\nu \rightarrow 0} c_1^{\nu}-n_1^{\nu} = 0.$$
With the same computation we can obtain that
$$ \lim_{\nu \rightarrow 0} c_2^{\nu}-n_2^{\nu} = 0.$$
This formally shows the convergence of the RGFM toward the GFM.

\subsection{The scheme}

We aim to numerically solve the RGFM given by \eqref{eq:n1r}-\eqref{eq:c2r} with Neumann boundary conditions by a finite-volume method. The choice of the boundary condition is arbitrary because we are interested in the dynamics of the system whilst the densities have not reached the boundary. To facilitate the comprehension, we omit the indices $\nu$. For the sake of simplicity, we only consider the 1D case on a finite interval $ \Omega = (a,b)$.

We divide the computational domain into finite-volume cells $C_j =[x_{j-1/2} ,x_{j+1/2} ] $ of uniform size $\Delta x$ with $x_j =j\Delta x, j\in \{1,...,M_x\} $,  and 
$x_{j}= \frac{x_{j-1/2} +x_{j+1/2} }{2}$ so that 
$$ a= x_{1/2}  < x_{3/2} < ... <  x_{j-1/2} < x_{j+1/2} < ... < x_{M_x-1/2}  < x_{M_x+1/2} = b$$
and define the cell average of functions $n_1(t,x)$ and $n_2(t,x)$ on the cell $C_j$ by
$$ \bar{n}_{\beta_j}(t)= \frac{1}{\Delta x} \int_{C_j}n_\beta(t,x)\,dx,\quad  \beta \in \{1,2\}.$$
A semi-discrete finite-volume scheme is obtained by integrating system \eqref{eq:n1r}-\eqref{eq:c2r} over $C_j$ and is given by
\begin{equation} \label{eq:scheme}
\begin{aligned}	
&\partial_t {\bar{n}}_{1_j}(t) = - \frac{F_{1,j+1/2}(t)-F_{1,j-1/2}(t)}{\Delta x}+{\bar{n}}_{1_j}(t) G_1(p_{1_j}(t)), \\
&\partial_t {\bar{n}}_{2_j}(t) = - \frac{F_{2,j+1/2}(t)-F_{2,j-1/2}(t)}{\Delta x}+{\bar{n}}_{2_j}(t) G_2({p}_{2_j}(t)),\\
& \alpha \frac{{c}_{1_{j+1}}(t)-2{c}_{1_{j}}(t)+{c}_{1_{j-1}}(t)}{(\Delta x)^2}-\frac{1}{\alpha} ({c}_{1_j}(t)-{\bar{n}}_{1_j}(t)) =0,\\
&\alpha \frac{{c}_{2_{j+1}}(t)-2{c}_{2_{j}}(t)+{c}_{2_{j-1}}(t)}{(\Delta x)^2}  - \frac{1}{\alpha} ({c}_{2_j}(t)-{\bar{n}}_{2_j}(t)) =0.
\end{aligned}	
\end{equation}
Here, $c_{\beta_{j}}(t) \approx c_\beta(t,x_j),\ \beta \in \{1,2\}$ and $F_{\beta,j+1/2}$ are numerical fluxes approximating 
$-n_{\beta} u_{\beta}:=-n_{\beta} \partial_x(p_{\beta}-(c_{\beta}-n_{\beta}))$ and defined by:
	\begin{equation} \label{eq:F}
	 F_{\beta,j+1/2} =(u_{{\beta}_{j+1/2}})^+\bar{n}_{\beta_j} +(u_{{\beta}_{j+1/2}})^- {\bar{n}}_{\beta_{j+1}},\quad \beta \in \{1,2\},
	 \end{equation}
where
	\begin{equation} \label{eq:u}
	 {u_\beta}_{j+1/2} = - \frac{\big(p_{\beta_{j+1}}-(c_{\beta_{j+1}}-{\bar n}_{\beta_{j+1}} )\big)-\big(p_{\beta_j}-(c_{\beta_j}-{\bar n}_{\beta_j} )\big)}{\Delta x},
	\end{equation}
and $(u_{{\beta}_{j+1/2}})^+=\max(u_{{\beta}_{j+1/2}},0)$ and $(u_{{\beta}_{j+1/2}})^-=\min(u_{{\beta}_{j+1/2}},0)$ are its positive and negative part, respectively. 
		
It is easy to see that scheme \eqref{eq:scheme}--\eqref{eq:u} is first order in space and if one uses an explicit Euler method for time evolution, then the scheme is also first order in time. Higher
order approximations can also be obtained, see, e.g., \cite{CCH}.

\subsection{CFL condition and properties of the scheme}
In this section, we present some fundamental properties the scheme is endowed with. In particular, we prove that for all times $t\ge 0$, scheme \eqref{eq:scheme}--\eqref{eq:u} preserves the non-negativity of the computed densities ${\bar{n}}_{1}$ and ${\bar{n}}_{2}$ and also ensures that the total density $n$ stays below $1$. The former property
guarantees physically meaningful values of the computed densities while the second one is enforced to make sense of the pressure law for which 1 is a singular point. 

We consider the semi-discrete finite volume scheme \eqref{eq:scheme}--\eqref{eq:u} and assume that is evolved in time by the forward Euler method. We denote by 
${\bar{n}}^{k}_{1_j}$, ${\bar{n}}^{k}_{2_j}$, $p^{k}_{1_j}$ and $p^{k}_{2_j}$ the computed densities and pressures obtained at time $t^k=k \Delta t$, i.e., ${\bar{n}}^{k}_{1_j}:={\bar{n}}_{1_j}(t^k)$,  
${\bar{n}}^{k}_{2_j}:={\bar{n}}_{2_j}(t^k)$, ${p}^{k}_{1_j}:={p}_{1_j}(t^k)$ and ${p}^{k}_{2_j}:={p}_{2_j}(t^k)$, and prove the following two propositions.
	
\begin{prop} (Positivity of the density)
Consider the semi-discrete finite volume scheme \eqref{eq:scheme}--\eqref{eq:u} that is evolved in time by the forward Euler method. Provided that the initial densities 
${\bar{n}}^{0}_{1_j} \geq 0$ and ${\bar{n}}^{0}_{2_j}\geq 0$ for all $j \in \{1,...,M_x\}$, and that the growth terms $G_{\beta}$ are nonnegative for $\beta=1,2$, a 
sufficient CFL condition for the  cell averages ${\bar{n}}^{k}_{1_j}$ and ${\bar{n}}^{k}_{2_j}$, with $k \in [0, \frac{N}{\Delta t}]$, to be positive is
\begin{equation} 
\Delta t \leq \frac{\Delta x}{2 \displaystyle \max_{\substack{j\in \{1,...,M_x\}, \\
\beta \in \{1,2\} } } \{({u^{k}_{\beta_{j+1/2}}})^+,-({u^{k}_{\beta_{j+1/2}}})^- \}}, \quad \forall k \in \left[0, \frac{N}{\Delta t}\right] ,\label{eq:CFL}
\end{equation}
where $({u^{k}_{\beta_{j+1/2}}})^+ := (u_{\beta_{j+1/2}})^+ (t^k)$ and $({u^{k}_{\beta_{j+1/2}}})^- := (u_{\beta_{j+1/2}})^- (t^k)$.
\end{prop} 
\begin{proof}
Assume that at a given time $t^k=k \Delta t$, ${\bar{n}}^{k}_{1_j} \geq 0$ and ${\bar{n}}^{k}_{2_j}\geq 0$ for all $j \in \{1,...,M_x\}$. Then, the new densities are given by the general formula
\begin{equation}
{\bar{n}}^{k+1}_{\beta_j}={\bar{n}}^{k}_{\beta_j}  - \frac{\Delta t }{\Delta x} (F^{k}_{\beta,j+1/2}-F^{k}_{\beta,j-1/2})+\Delta t ~{\bar{n}}^{k}_{\beta_j} G_1(p^k_{\beta_j})
\end{equation}
where $F^{k}_{\beta,j+1/2} := F_{\beta,j+1/2}(t^k)$.
Taking into account formula \eqref{eq:F} for fluxes $F^{k}_{\beta,j+1/2}$ and the fact that the growth terms $G_{\beta}$ are nonnegative for $\beta=1,2$, we obtain
\begin{equation*}
\begin{aligned}
{\bar{n}}^{k+1}_{\beta_j}&\geq {\bar{n}}^{k}_{\beta_j}-\frac{\Delta t}{\Delta x}\left[(u^{k}_{\beta_{j+1/2}})^+\,{\bar{n}}^{k}_{\beta_j} + 
(u^{k}_{\beta_{j+1/2}})^-\,{\bar{n}}^{k}_{\beta_{j+1}}-(u^{k}_{\beta_{j-1/2}})^+\,{\bar{n}}^{k}_{\beta_{j-1}}-(u^{k}_{\beta_{j-1/2}})^-\,{\bar{n}}^{k}_{\beta_j}\right] \\
&\geq \frac{\Delta t }{\Delta x} (u^{k}_{\beta_{j-1/2}})^+\,{\bar{n}}^{k}_{\beta_{j-1}} + 
\left(\frac{1}{2} - \frac{\Delta t }{\Delta x} (u^{k}_{\beta_{j+1/2}})^+\right) {\bar{n}}^{k}_{\beta_j} \\
& + \left(\frac{1}{2} + \frac{\Delta t }{\Delta x} (u^{k}_{\beta_{j-1/2}})^-\right)\,{\bar{n}}^{k}_{\beta_j}
-\frac{\Delta t }{\Delta x} (u^{k}_{\beta_{j+1/2}})^-\,{\bar{n}}^{k}_{\beta_{j+1}}.
\end{aligned}
\end{equation*}
By definition, $(u^{k}_{\beta_{j+1/2}})^-\leq 0$ and $(u^{k}_{\beta_{j+1/2}})^+\geq 0$ for all $j \in \{1,...,M_x\}$. Then, provided that the CFL condition \eqref{eq:CFL}  
is satisfied,
$$ 
\frac{1}{2} - \frac{\Delta t }{\Delta x} (u^{k}_{\beta_{j+1/2}})^+\geq 0 \quad \text{and} \quad 
\frac{1}{2} + \frac{\Delta t }{\Delta x} (u^{k}_{\beta_{j-1/2}})^- \geq 0 \quad \forall j \in \{1,...,M_x\}.$$
Since $({\bar{n}}^{k}_{1_j})_{j \in \{1,...,M_x\}},({\bar{n}}^{k}_{2_j})_{j \in \{1,...,M_x\}}$ are non-negative, we conclude that ${\bar{n}}^{k+1}_{1_j} \geq 0$ and ${\bar{n}}^{k+1}_{2_j}\geq 0$ for all $j \in \{1,...,M_x\}$.
\end{proof}

\begin{prop} (Maximum total density) Consider the semi-discrete finite volume scheme \eqref{eq:scheme}--\eqref{eq:u} that is evolved in time by the forward Euler method. Provided that the initial total density ${\bar{n}}^{0}_j={\bar{n}}^{0}_{1_j} +{\bar{n}}^{0}_{2_j} < 1$ for all $j \in \{1,...,M_x\}$, a sufficient CFL condition for the average total densities ${\bar{n}}^{k}_j = {\bar{n}}^{k}_{1_j}+{\bar{n}}^{k}_{2_j}$, to be below 1 for all $k \in [0, \frac{N}{\Delta t}]$ is
\begin{equation} 
\begin{split}
	\Delta t \leq (\frac{1}{\displaystyle \max_{\substack{j\in \{1,...,M_x\}, \\
                             \beta \in \{1,2\} } }  {\bar{n}}^{k}_{\beta_j}} -1) \frac{1}{\frac{4}{\Delta x} \displaystyle \max_{\substack{j\in \{1,...,M_x\}, \\
                             \beta \in \{1,2\} } } \{({u^{k}_{\beta_{j+1/2}}})^+,-({u^{k}_{\beta_{j+1/2}}})^+ \} +\displaystyle \max_{\substack{j\in \{1,...,M_x\}, \\
                             \beta \in \{1,2\} } } G_\beta({p^k_\beta}_j) },  \\
                             \quad \forall k \in [0, \frac{N}{\Delta t}] . \label{eq:CFL2}
                             \end{split}
	\end{equation}
\end{prop} 

\begin{proof}
Assume that at a given time $t^k=k \Delta t$, ${\bar{n}}^{k}_{1_j} \leq 1$ and ${\bar{n}}^{k}_{2_j}\leq 1$ for all $j \in \{1,...,M_x\}$. The densities 
${\bar{n}}^{k+1}_{\beta_j},\ \beta=\{1,2\}$, at the following time step are given by 
\begin{equation*}
\begin{aligned}
{\bar{n}}^{k+1}_{\beta_j} &= {\bar{n}}^{k}_{\beta_j}  - \frac{\Delta t }{\Delta x} (F^{k}_{\beta,j+1/2}-F^{k}_{\beta,j-1/2})
+\Delta t ~{\bar{n}}^{k}_{\beta_j} G_1(p^k_{\beta_j}) \\
&\leq {\bar{n}}^{k}_{\beta_j}  - \frac{\Delta t}{\Delta x}\left[(u^{k}_{\beta_{j+1/2}})^+\,{\bar{n}}^{k}_{\beta_j} + 
(u^{k}_{\beta_{j+1/2}})^-\,{\bar{n}}^{k}_{\beta_{j+1}}-(u^{k}_{\beta_{j-1/2}})^+\,{\bar{n}}^{k}_{\beta_{j-1}}-(u^{k}_{\beta_{j-1/2}})^-\,{\bar{n}}^{k}_{\beta_j}\right] \\[1ex]
&\hspace*{1.2cm}+\Delta t ~{\bar{n}}^{k}_{\beta_j} G_1(p^k_{\beta_j})  \\
&\leq \Big(1+4 \frac{\Delta t }{\Delta x} \displaystyle \max_{j\in \{1,...,M_x\}} \{(u^{k}_{\beta_{j+1/2}})^+,-(u^{k}_{\beta_{j+1/2}})^- \}
+ \Delta t  \displaystyle \max_{j\in \{1,...,M_x\}} G_\beta(p^k_{\beta_j} \Big) \displaystyle \max_{j\in \{1,...,M_x\} }  {\bar{n}}^{k}_{\beta_j}.
\end{aligned}
\end{equation*} 
Adding the last inequality for the two densities we get that
\begin{equation*}
\begin{aligned}
{\bar{n}}^{k+1}_{j} &\leq\Big(1+4 \frac{\Delta t }{\Delta x} \displaystyle \max_{\substack{j\in \{1,...,M_x\}, \\ \beta \in \{1,2\}}} 
\{(u^{k}_{\beta_{j+1/2}})^+,-(u^{k}_{\beta_{j+1/2}})^- \} \\
&\hspace*{1.1cm}+ \Delta t  \displaystyle \max_{\substack{j\in \{1,...,M_x\}, \\
\beta \in \{1,2\} } } G_\beta(p^k_{\beta_j}) \Big) \displaystyle \max_{\substack{j\in \{1,...,M_x\}, \\ \beta \in \{1,2\} } }  {\bar{n}}^{k}_{\beta_j}.
\end{aligned}
\end{equation*} 
	Then a sufficient condition for the total density to stay below 1 is \eqref{eq:CFL2}.
\end{proof}

\begin{remark}
Eqs. \eqref{eq:CFL} and  \eqref{eq:CFL2} give the conditions to choose the time step in the numerical simulations. It should be however observed that in some computations we needed to reduce the time step in order to avoid oscillations the pressure develops when the density is close to the singular point $n=1$.  
\end{remark}
\begin{remark}
Similar theorem can also be proven if the second-order upwind spatial scheme from \cite{CCH} is used and the forward Euler method is replaced by a higher-order
SSP ODE solver since a time step in such solvers can be written as a convex combination of several forward Euler steps, see, e.g., \cite{GKS,GST}.
\end{remark}

\subsection{Numerical simulation of the GFM with initial mixing}
Finally, we illustrate the performance of RGFM in a number of numerical examples when the initial populations are not segregated. 
We take as growth functions, 
\begin{align} \label{ex3:G} G_1(p)=(20-p) \quad \mbox{and} \quad G_2(p)=(10-p). \end{align}
The numerical parameter values are $\epsilon =0.01$, $m =10$, $\alpha =0.01$, $\nu =0.001$, and the initial densities are given by 
\begin{align} \label{ex3:n} n_1^{\rm ini}(x) = 0.7e^{-5x^2} \quad \mbox{and} \quad n_2^{\rm ini}(x) = 0.5e^{-5(x-0.5)^2}+0.6e^{-5(x+1)^2}. \end{align}
 These initial conditions depicted in Fig.~\ref{fig3}(a) have been chosen to show the evolution with non segregated and non symmetric initial conditions. Fig.~\ref{fig3} (b), (c), (d) represent the solution at time $t=0.01$, $t=0.1$ and $t=0.15$, respectively. The red line represents the species $n_1$ while the blue line represents the species $n_2$.

As expected, the two species are growing and diffusing. Since the initial overlap is quite strong, segregation does not appear immediately. In Fig.~\ref{fig3}(b), we see that the red species, which is growing faster than the blue one, also grows in the interval occupied by the blue species. When the species reach the maximum density in Fig.~\ref{fig3}(c), interfaces between the two populations are created. We can observe that the red species is split into two groups. In Fig.~\ref{fig3}(d), we also observe that the inner species is pushing the other one to have more space to grow.

  \begin{figure}[!h]
   \centering
\includegraphics[width=1\linewidth]{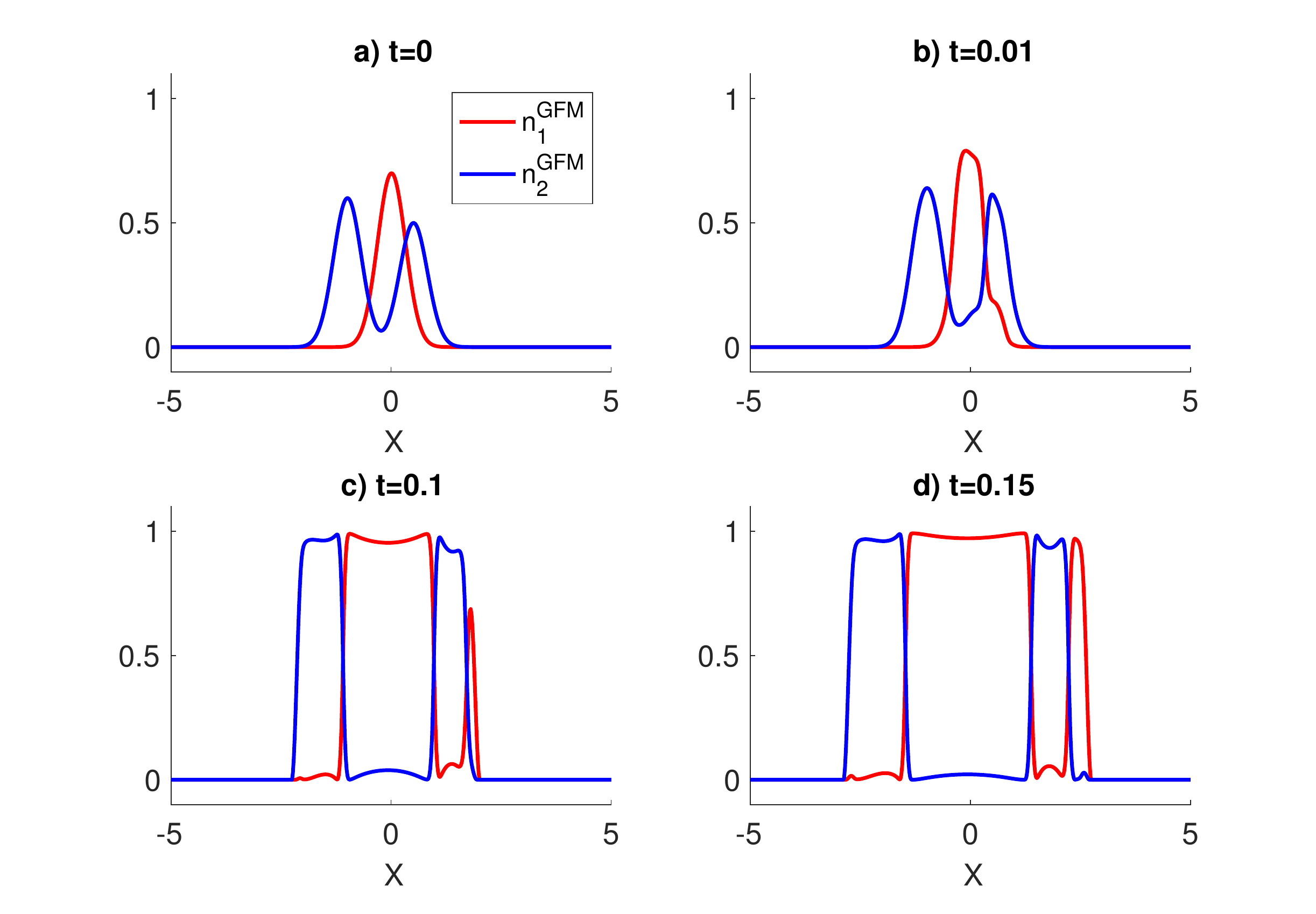}
\caption{Densities $n_1$ (red), $n_2$ (blue) as functions of position $x$ for the GFM at different times: (a) $t=0$ , (b)  $t=0.01$, (c) $t=0.1$, (d) $t=0.15$. Initial conditions: density of the GFM from Eq.~\eqref{ex3:n}, growth function from Eq.~\eqref{ex3:G}.}
\label{fig3}
\end{figure}

However we still observe some mixing. On the one hand, looking at the central red species, we observe that the density of the blue species is not exactly equal to 0. This is due to the low value of the parameter $m$ which is the exponent of the repulsion pressure. As $m$ becomes larger, the simulation fulfills the segregation property faster. This can be observed on Fig.~\ref{fig4}, where we plot results obtained from the simulation with the same initial condition than previously given by Eqs. \eqref{ex3:G} and \eqref{ex3:n} at time $t=0.15$ with parameters $\epsilon=0.1$, $\alpha=0.01$, $\nu=0.001$. Respectively in Fig.~\ref{fig4} (a), (b), (c) the densities are plotted for the cases $ m=5$, $m=10$ and $m=50$. We can observe in Fig.~\ref{fig4}(c) that there is more segregation than in Figs.~\ref{fig4}(a) and (b). If we compare the red central domain, in Fig.~\ref{fig4}(c) the domain is smaller than in Fig.~\ref{fig4}(b) because the blue species has grown in the middle of it in order to fulfill the non-mixing property. Then in Fig.~\ref{fig4}(c), we don't observe the small bump of the blue population found in Fig.~\ref{fig4}(b).

  \begin{figure}[!h]
   \centering
\includegraphics[width=1\linewidth]{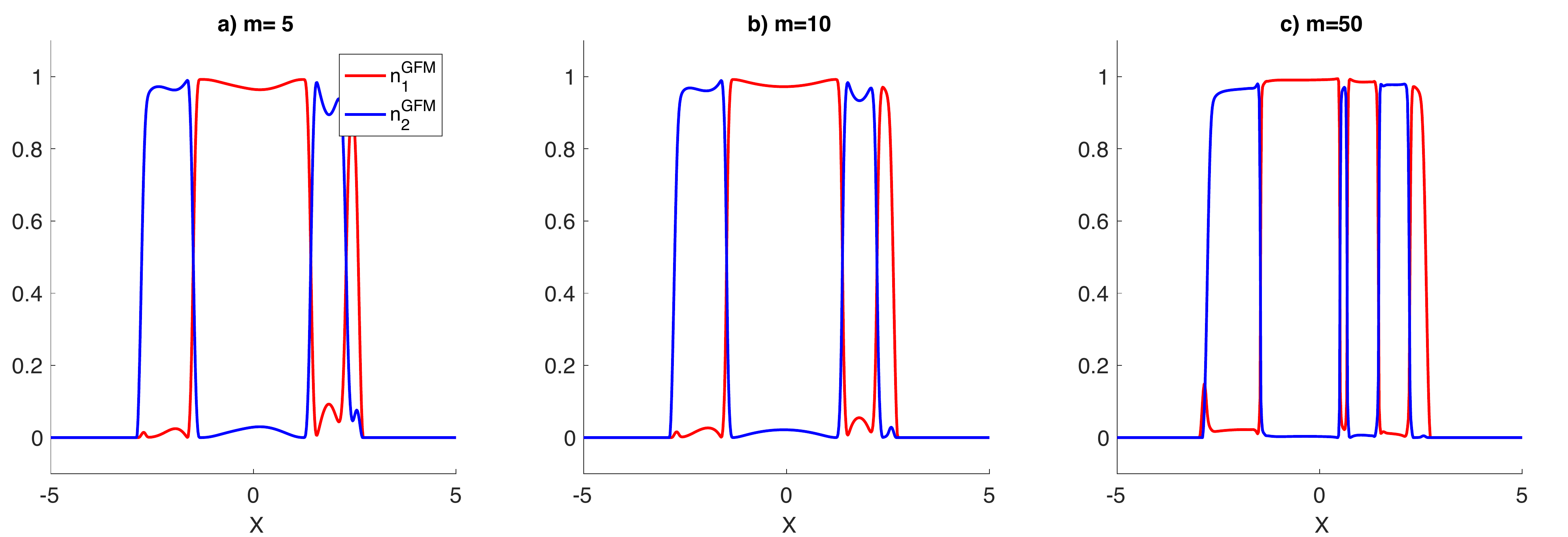}
\caption{Densities $n_1$ (red), $n_2$ (blue) as functions of position $x$ for the GFM  for different values of $m$: (a) $m=5$ , (b)  $m=10$, (c) $m=50$. Initial conditions: density of the GFM from Eq.~\eqref{ex3:n}, growth function from Eq.~\eqref{ex3:G}.}
\label{fig4}
\end{figure}

On the other hand, at the interface, the species mix in a small interval. This is due to the finiteness of the parameter $\alpha$ which multiplies the fourth order diffusion term and of $\nu$ which is the relaxation parameter. We observe these phenomena in Fig.~\ref{fig5}, where we show the results obtained from simulations done with the same initial conditions \eqref{ex3:G}, \eqref{ex3:n} than previously, at time $t=0.15$ with parameters $\epsilon=0.1$, $m=10$. Respectively in panels Fig.~\ref{fig5} (a), (b), (c) the densities are plotted for the cases $ \alpha=0.1, \nu =0.001$; $ \alpha=0.01, \nu =0.001$ and $ \alpha=0.001, \nu =0.0001$ (we have observed that the parameter $\nu$ needs to be smaller than $\alpha$). We can observe that as $\alpha$ becomes smaller, the width of the interface between the two populations gets smaller. We can also observe in Fig.~\ref{fig5}(c) that on the right side, thin stripes with alternating populations appear. This is the kind of dynamics we observed in simulations performed with $\alpha=0$. Indeed without the fourth order term in the GFM, the numerical simulations do not have the same dynamics as the HSM. This has been verified with numerical schemes of a similar structure to the one exposed in this paper, as well as with schemes derived from the gradient flow in the Wasserstein metric structure of the system. It appears that without the fourth order term, a species surrounded by an other one will prefer to split into two small subdomains crossing through the other species in order to grow instead of pushing it, which does not correspond to what we find with the HSM dynamics in the incompressible limit. For these reasons, the fourth order term is essential to producing realistic dynamics. What this analysis shows is that the choice of the parameter $\alpha$ is critical. In particular, in order to recover the HSM, $\alpha$ has to be made smaller as $\epsilon \rightarrow 0$ and $m \rightarrow \infty$. The values of the parameters $\epsilon$, $m$, $\alpha$ must be carefully chosen to match a particular application.

\begin{figure}[!h]
   \centering
\includegraphics[width=1.\linewidth]{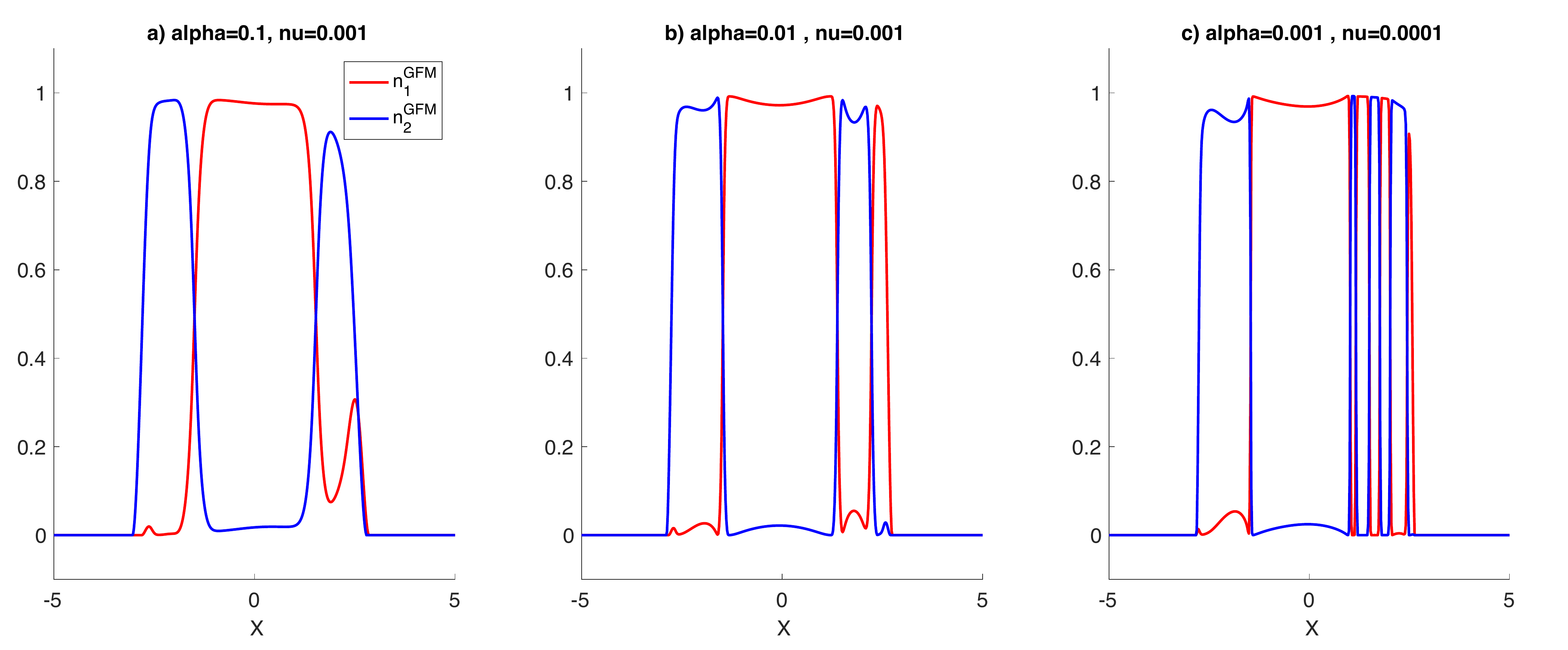}
\caption{Densities $n_1$ (red), $n_2$ (blue) as functions of position $x$ for the GFM for different values of $\alpha$ and $\nu$: (a) $ \alpha=0.1, \nu =0.001$, (b) $ \alpha=0.01, \nu =0.001$, (c) $ \alpha=0.001, \nu =0.0001$. Initial conditions: density of the GFM from Eq.~\eqref{ex3:n}, growth function from Eq.~\eqref{ex3:G}.}
\label{fig5}
\end{figure}

\subsection{Conclusion}

In this paper, we have presented a gradient flow model for two populations which avoid mixing. In addition we have introduced a numerical scheme and used it to study the behaviour of the system with different parameters. This model is a generalisation of a single population case which has been studied in the literature previously. 

This paper demonstrates by a combination of analytical and numerical arguments, that the gradient flow model converges to a Hele-Shaw free interface/boundary model, when an appropriate set of parameters is sent to zero (or infinity). The analytical proof is only formal and is supported by numerical simulations. In particular, we observe that the speed of the boundaries and interfaces of the gradient flow model for parameters taken in the asymptotic regime is the same as those computed by the Hele-Shaw model. This is verified with values of the parameters fairly far from the asymptotic regime which means that the convergence is quite fast.

Perspectives for this work are both on the analytical and numerical sides. On the analytical side, a rigorous proof of the convergence of the gradient flow model to the Hele-Shaw one seems within reach in the case $\alpha =0$ and $q=0$. On the numerical side, simulations of two-dimensional cases will be developed and applied to the modelling of tissue growth and growth termination.

 
\end{document}